\pdfoutput=1

\documentclass{article}
\usepackage{kr}

\usepackage{times}
\usepackage{soul}
\usepackage{url}
\usepackage[hidelinks]{hyperref}
\usepackage[utf8]{inputenc}
\usepackage[small]{caption}
\usepackage{graphicx}
\usepackage{amsmath,amsfonts,amssymb}
\usepackage{amsthm}
\usepackage{booktabs}
\usepackage{algorithm}
\usepackage{algorithmic}
\usepackage{xcolor}
\usepackage{thm-restate}

\usepackage{bm}

\newtheorem{example}{Example}

\newtheorem{definition}{Definition}

\title{Combining Global and Local Merges in Logic-based Entity Resolution}

\author{%
Meghyn Bienvenu$^1$\and
Gianluca Cima$^2$\and
Víctor Gutiérrez-Basulto$^{3}$\and
Yazmín Ibáñez-García$^3$
\affiliations
$^1$Univ. Bordeaux, CNRS, Bordeaux INP, LaBRI, UMR 5800\\
$^2$Department of Computer, Control and Management Engineering, Sapienza University of Rome\\
$^3$School of Computer Science \& Informatics, Cardiff University\\
\emails
meghyn.bienvenu@labri.fr, cima@diag.uniroma1.it, \{gutierrezbasultov, ibanezgarciay\}@cardiff.ac.uk
}

\usepackage{mathtools}
\usepackage{xspace}
\usepackage{amsfonts}
\usepackage{multirow}
\usepackage{subcaption}
\usepackage{array}
\newcommand{\typef}{\mathbf{type}}
\def\dcons{\mathsf{Dom}}
\def\vars{\mathsf{vars}}
\def\qcons{\mathsf{cons}}
\def\ratom{\pi}

\newcommand{\StrPos}{\mathsf{Cells}}

\def\harrow{\Rightarrow}
\def\sarrow{\dashrightarrow}

\newcommand{\E}{\Sigma}
\newcommand{\HRO}{\Gamma_h^o}
\newcommand{\SRO}{\Gamma_s^o}

\def\hsrulesO{\Gamma_{O}}
\def\hsrulesV{\Gamma_{V}}
\newcommand{\sro}{\sigma^o}
\newcommand{\hro}{\rho^o}
\newcommand{\merO}{E}
\newcommand{\merV}{V}
\newcommand{\extdat}[3]{{#1}_{{#2},{#3}}}

\def\eqrel{\mathsf{EqRel}}

\def\existpb{\textsc{Existence}\xspace}
\def\recpb{\textsc{Rec}\xspace}
\def\maxrecpb{\textsc{MaxRec}\xspace}
\def\cmergepb{\textsc{CertMerge}\xspace}
\def\pmergepb{\textsc{PossMerge}\xspace}
\def\certanspb{\textsc{CertAns}\xspace}

\def\possanspb{\textsc{PossAns}\xspace}

\def\ex{\mathsf{ex}}

\newcommand{\Author}{\text{Author}}
\newcommand{\Paper}{\text{Paper}}
\newcommand{\Wrote}{\text{Wrote}}

\newcommand{\MaxSOL}{\mathsf{MaxSol}}
\newcommand{\SOL}{\mathsf{Sol}}

\newcommand{\SRV}{\Gamma_s^v}
\newcommand{\HRV}{\Gamma_h^v}

\newcommand{\hrv}{\rho^v}
\newcommand{\Obj}{{\bm{\mathsf{O}}}}
\newcommand{\ObjD}{\mathsf{Obj}}
\newcommand{\Val}{{\bm{\mathsf{V}}}}
\newcommand{\tidD}{{\bm{\mathsf{TID}}}}
\newcommand{\tid}{\mathsf{tid}}

\newcommand{\linkO}{\mathsf{EqO}}
\newcommand{\linkV}{\mathsf{EqV}}

\newcommand{\dom}[1]{\mathsf{dom}(#1)}

\newcommand{\PTIME}{\textsc{P}\xspace}

\newcommand{\NP}{\textrm{NP}\xspace}

\newcommand{\coNP}{\textrm{coNP}\xspace}

\newcommand{\PItwop}{{\Pi_2^p}\xspace}
\newcommand{\SItwop}{{\Sigma_2^p}\xspace}

\newcommand{\myi}{(\emph{i})\xspace}
\newcommand{\myii}{(\emph{ii})\xspace}
\newcommand{\myiii}{(\emph{iii})\xspace}

\renewcommand{\S}{\mathcal{S}}

\newcommand{\ie}{i.e.~}
\newcommand{\eg}{e.g.~}
\newcommand{\resp}{resp.~}

\newcommand{\tup}[1]{\langle #1\rangle}
\newcommand{\per}{\mbox{\bf .}}
\newcommand{\true}{\texttt{true}\xspace}
\newcommand{\false}{\texttt{false}\xspace}

\newtheorem{remark}{Remark}
\newtheorem{lemma}{Lemma}

\makeatletter
\newcommand{\thickhline}{%
    \noalign {\ifnum 0=`}\fi \hrule height 1pt
    \futurelet \reserved@a \@xhline
}
\newcolumntype{"}{@{\hskip\tabcolsep\vrule width 1pt\hskip\tabcolsep}}
\makeatother

\makeatletter
\newcommand{\thinhline}{%
    \noalign {\ifnum 0=`}\fi \hrule height 0.2pt
    \futurelet \reserved@a \@xhline
}
\newcolumntype{"}{@{\hskip\tabcolsep\vrule width 1pt\hskip\tabcolsep}}
\makeatother

\newcolumntype{?}{!{\vrule width 1.1pt}}

\begin{document}

\maketitle

\begin{abstract}
In the recently proposed \textsc{Lace} framework for collective entity resolution, logical 
rules and constraints are used to identify pairs of entity references (\eg author or paper ids) that denote 
the same entity. This identification is global:  all occurrences of 
those entity references (possibly across multiple database tuples) are deemed equal 
and can be merged.
By contrast, a local form of merge is often more natural when identifying pairs of data
values, \eg some occurrences of `J. Smith' may be equated with `Joe Smith', while 
others should merge with `Jane Smith'. 
This motivates us to extend \textsc{Lace} with local merges of values and explore the computational properties of the resulting formalism. 
\end{abstract}

\section{Introduction}\label{sec:Intro}

 Entity resolution (ER) is a data quality management task aiming  at identifying database different constants (of the same type) 
 that refer to the same real-world entity~\cite{SinglaICDM06}. Given the fundamental nature of this problem, several variants of ER (also known as record linkage or deduplication) have been investigated. \emph{Collective} entity resolution~\cite{BhattacharyaTKDD07,DengICDE22})
 considers the joint resolution (match, merge) of entity references or values of multiple types across multiple tables,  
 e.g.\ using the merge of two authors to infer that two paper ids have to be merged as well. 
Various approaches to collective ER, 
with different formal foundations, have been developed: probabilistic approaches, deep learning approaches, and approaches based on rules and constraints, see~\cite{ChristophidesACMSurvey21} for a survey. 

We have recently proposed  \textsc{Lace}~\cite{BienvenuPODS22}, a declarative framework for collective ER based upon
logical rules and constraints. \textsc{Lace} employs hard and soft rules to define mandatory and possible merges.
The semantics of \textsc{Lace} is dynamic: ER solutions are generated by sequences of rule applications, where rules are evaluated over the current induced database, taking into account all previously derived merges.
This makes it possible to support recursive scenarios (e.g.\ a merge of authors triggers a merge of papers which in turn enables another merge of authors), while ensuring that all merges have a (non-circular) derivation.  The semantics is also global in the sense that 
\emph{all} occurrences of the matched constants are merged, rather than only those constant occurrences used in deriving the match. Such a global semantics is well suited 
for merging constants that are entity references (e.g.\ authors or paper ids)
and has been used in other prominent logic-based approaches \cite{ArasuICDE09,BurdickTODS2016}. However, for merging attribute values (e.g.\ author names), a local semantics, which considers the context in which a value occurs, is more appropriate. Indeed, a local semantics allows some occurrences of  ‘J.\ Smith’ to be matched to ‘Joe Smith’ and others to ‘Jane Smith’, without (wrongly) equating the latter two constants. 
\emph{Matching dependencies}~\cite{BertossiTCS13,FanPODS08,FanVLDB09} are an example of a principled logical formalism for 
merging values. 

To the best of our knowledge, there is currently no 
ER framework that supports both global and local merges. 
This motivates us to introduce \textsc{Lace}$^{\tiny \text{+}}$, an extension of \textsc{Lace} with local merges of values, in which 
local merges may enable global merges, and vice versa. In particular, local merges can
resolve constraint violations which would otherwise block desirable global merges. \textsc{Lace}$^{\tiny \text{+}}$ 
extends \textsc{Lace}'s syntax by adding hard and soft rules for values,  
but it departs from \textsc{Lace} semantics by considering sets of constants, rather than single constants, as arguments in induced databases. 
Intuitively, such a set of constants 
provides alternative representations of the same information, e.g. different forms of a name. 
The semantic treatment of local merges within \textsc{Lace}$^{\tiny \text{+}}$ aligns with the design of the generic ER framework 
Swoosh~\cite{BenjellounVLDBJ09}.

 Our main contributions are the introduction of the new \textsc{Lace}$^{\tiny \text{+}}$ framework and the exploration of its computational properties. Our complexity analysis shows 
that the addition of local merges does not increase the data complexity of the considered reasoning tasks. 
We also show how an existing answer set programming (ASP) encoding of ER solutions in \textsc{Lace} can be extended to handle local merges of values.

For a discussion of related work, see \cite{BienvenuPODS22}, and for an extension of \textsc{Lace} with repairing, see~\cite{BienvenuIJCAI23}.  



\section{Preliminaries}\label{sec:Preliminaries}
\noindent \textbf{Databases} We assume that \emph{constants} are drawn from three 
infinite and pairwise disjoint
sets: a set $\Obj$ of \emph{object constants} (or \emph{objects}), serving as references to real-world entities (e.g.\ paper and author ids), 
a set $\Val$ of \emph{value constants} (or \emph{values}) from the considered datatypes (e.g.\ strings for names of authors and paper titles, dates for time of publication),
and a set $\tidD$ of  
 \emph{tuple identifiers (tids)}. 
A \emph{(database) schema} 
$\S$ consists of a finite set of \emph{relation symbols}, each having an associated
arity $k \in \mathbb{N}$ and type vector $\{\Obj,\Val\}^k$. 
We use $R/k \in \S$ to indicate that the relation symbol $R$ from $\S$ has arity $k$,
and denote by $\typef(R,i)$ the $i$th element of $R$'s type vector. If $\typef(R,i) = \Obj$ (resp.\ $\Val$),
we call $i$ an \emph{object (resp.\ value) position} of $R$. 
A \emph{($\tidD$-annotated) $\S$-database} 
is a finite set $D$ of \emph{facts} of the form $R(t,c_1, \ldots, c_k)$, where $R/k \in \S$, $t \in \tidD$, and $c_i \in \typef(R,i)$ for every $1 \leq i \leq k$.
We require that each $t \in \tidD$ occurs in at most one fact of $D$. 
We say that  $t$ (resp.\ $c_i$) occurs in position $0$ (resp.\ $i \in \{1, \ldots, k\}$) of 
$R(t,c_1, \ldots, c_k)$,
and slightly abusing notation, use $t$ and $t[j]$ respectively to refer to the unique fact having tid $t$, and to the constant in the $j$th position of that fact. 
%
The set of constants (resp.\ objects) occurring in $D$ is denoted $\dcons(D)$ (resp.\ $\ObjD(D)$),
and the 
set $\StrPos(D)$ of 
 \emph{(value) cells} of $D$  
is defined as $\{\tup{t,i} \mid R(t,c_1, \ldots, c_k) \in D, \typef(R,i) = \Val \}$.\smallskip





\noindent \textbf{Queries} In the setting of $\tidD$-annotated $\S$-databases,
a \emph{conjunctive query (CQ)}  has the form $q(\vec{x})=\exists \vec{y} \per \varphi(\vec{x},\vec{y})$, where $\vec{x}$ and $\vec{y}$ are disjoint tuples of variables, and $\varphi(\vec{x},\vec{y})$ is a conjunction of relational atoms of the form $R(u_0,u_1, \ldots, u_k)$, where $R/k \in \S$ and $u_i \in \Obj \cup \Val \cup \tidD \cup \vec{x}\cup \vec{y}$ for $0 \leq i \leq k$. 
When formulating entity resolution rules and constraints, we shall also consider extended forms of CQs that 
may contain inequality atoms or atoms built from a set of 
binary \emph{similarity predicates}. 
Note that such atoms will not contain the tid position and have a fixed meaning\footnote{The extension of similarity predicates is typically defined by applying some similarity metric, \eg edit distance, and keeping those pairs of values whose score exceeds a given threshold.}.
As usual, the \emph{arity} of $q(\vec{x})$
is the length of $\vec{x}$, and queries of arity 0 are called \emph{Boolean}. Given an $n$-ary query $q(x_1, \ldots, x_n)$ and $n$-tuple of constants $\vec{c}=(c_1, \ldots, c_n)$, we denote by $q[\vec{c}]$ the Boolean query obtained by replacing each $x_i$ by $c_i$. We use \emph{$\vars(q)$} (resp.\ $\qcons(q)$) for the set of variables (resp.\ constants) in $q$. 

\smallskip

\noindent \textbf{Constraints} Our framework will also employ denial constraints (DCs)~\cite{BertossiMC2011,FanMC2012}. 
A \emph{denial constraint} over a schema $\S$ 
takes the form 
$
    \exists \vec{y} \per \varphi(\vec{y}) \rightarrow \bot,
$
where $\varphi(\vec{y})$ is a Boolean CQ with inequalities, whose relational atoms use relation symbols from $\S$. 
We impose the standard safety condition: each variable 
occurring in an inequality atom must also occur in some relational atom. 
Denial constraints notably generalize the well-known class of \emph{functional dependencies (FDs)}. 
To simplify the presentation, we sometimes omit the initial quantifiers from DCs. 

\noindent \textbf{Equivalence Relations} We recall that an \emph{equivalence relation} on a set $S$
is a binary relation on $S$ that is reflexive, symmetric, and transitive. We use 
$\eqrel(P,S)$ for the smallest equivalence relation on $S$ that extends $P$.


\section{\textsc{Lace}$^{\tiny \text{+}}$ Framework}\label{sec:Framework}
This section presents and illustrates \textsc{Lace}$^{\tiny \text{+}}$, an extension of the \textsc{Lace} framework to handle local merges 
of values. 

\subsection{Syntax of \textsc{Lace}$^{\tiny \text{+}}$ 
Specifications}
As in \textsc{Lace}, 
we consider \emph{hard and soft rules for objects (over schema $\S$)}, 
which take respectively the forms:
$$
    q(x,y) \harrow \linkO(x,y) \quad q(x,y) \sarrow \linkO(x,y)
$$ 
where $q(x,y)$ is a CQ whose atoms may use relation symbols from $\S$ as well as similarity predicates
and whose free variables $x$ and $y$ occur only in object positions. Intuitively, the above hard (resp.\ soft) rule states that 
$(o_1,o_2)$ being an answer to $q$ provides sufficient (resp.\ reasonable) evidence for concluding that 
$o_1$ and 
$o_2$ 
refer to the same real-world entity. The special relation symbol $\linkO$ (not in $\S$) is used to store 
such merged pairs of object constants. 

%

To handle local identifications of values, we introduce 
\emph{hard and soft rules for values (over $\S$)}, which take the forms:  
\begin{align*}
q(x_t,y_t) \harrow&\, \linkV(\tup{x_t,i},\tup{y_t,j}) \\ q(x_t,y_t) \sarrow &\, \linkV(\tup{x_t,i},\tup{y_t,j})
\end{align*} 
where $q(x_t,y_t)$ is a CQ whose atoms may use relation symbols from $\S$ as well as similarity predicates, 
variables $x_t$ and $y_t$ each occur once in $q$ in position $0$ of (not necessarily distinct) relational atoms with relations $R_x \in \S$ and $R_y \in \S$, respectively, 
and $i$ and $j$ are value positions of $R_x$ and $R_y$, respectively.
Intuitively, such a hard (resp.\ soft) rule states that a pair of tids $(t_1,t_2)$ being an answer to $q$ provides sufficient (resp.\ reasonable) evidence for concluding that the
values in cells $\tup{x_t,i}$ and $\tup{y_t,j}$ 
are non-identical representations of the same information. 
The special relation symbol $\linkV$ (not in $\S$ and distinct from $\linkO$) is used to store pairs of value cells which have been merged. 



%

\begin{definition}~\label{def:ER}
    A  \textsc{Lace}$^{\tiny \text{+}}$ \emph{entity resolution (ER) specification} $\E$ for schema $\S$ takes the form $\E=\tup{\hsrulesO,\hsrulesV,\Delta}$, 
    where $\hsrulesO = \HRO \cup \SRO$ is a finite set of hard and soft rules for objects, 
    $\hsrulesV=\HRV \cup \SRV$ is a finite set of hard and soft rules for values, 
    and $\Delta$ is a finite set of denial constraints, all over $\S$.
\end{definition}

\begin{example}
The schema $\S_{\ex}$, database $D_{\ex}$, and ER specification $\E_\ex=\langle \Gamma^{O}_\ex, \Gamma^{V}_\ex, \Delta_\ex \rangle$
of our running example are given in Figure~\ref{fig:mainEx}.  Informally, the denial constraint $\delta_1$ is an FD
saying that an author id is associated with at most one 
author name, while the constraint $\delta_2$ forbids the existence of a paper written by the chair of the conference in which the paper was published. The hard rule 
$\hro_1$ states that if two author ids have the same name and the same institution, then they refer to the same author. The soft rule 
$\sro_1$ states that authors who wrote a paper in common and have similar names are likely to be the same. 
  Finally, the hard rule 
$\hrv_1$ locally merges similar names associated with the same author id.
\end{example}

\begin{figure*}[!htb]
    \begin{subtable}[t]{0.315\textwidth}
        \caption*{Author($\tid$,~aid,~name,~inst)}
        \begin{tabular}{| c ? c | c | c |}
            \thinhline
            $\tid$&aid&name&inst\\
            \thickhline
            $t_1$&$a_1$&J.~Smith&Sapienza\\
            \thinhline
            $t_2$&$a_2$&Joe Smith&Oxford\\
            \thinhline
            $t_3$&$a_3$&J.~Smith&NYU\\
            \thinhline
            $t_4$&$a_4$&Joe Smith&NYU\\
            \thinhline
            $t_5$&$a_5$&Joe Smith&Sapienza\\
            \thinhline
            $t_6$&$a_6$&Min Lee&CNRS\\
            \thinhline
            $t_7$&$a_7$&M.~Lee&UTokyo\\
            \thinhline
            $t_8$&$a_8$&Myriam Lee&Cardiff\\
            \thinhline
        \end{tabular}
    \end{subtable}
    \hfill
    \begin{subtable}[t]{0.49\textwidth}
        \centering
        \caption*{Paper($\tid$,~pid,~title,~conf,~ch)}
        \begin{tabular}{| c ? c | c | c | c |}
            \thinhline
            $\tid$&pid&title&conf&ch\\
            \thickhline
            $t_9$&$p_1$&Logical Framework for ER&PODS'21&$a_6$\\
            \thinhline
            $t_{10}$&$p_2$&Rule-based approach to ER&ICDE'19&$a_4$\\
            \thinhline
            $t_{11}$&$p_3$&Query Answering over DLs&KR'22&$a_1$\\
            \thinhline
            $t_{12}$&$p_4$&CQA over DL Ontologies&IJCAI'21&$a_1$\\
            \thinhline
            $t_{13}$&$p_5$&Semantic Data Integration&AAAI'22&$a_8$\\
            \thinhline
        \end{tabular}
        \vspace{4pt}

        The sim predicate $\approx$ is such that the names of authors $a_1$, $a_2$, $a_3$, $a_4$, and $a_5$ are pairwise similar, and both the names of authors $a_6$ and $a_8$ are similar to the name of author $a_7$
    \end{subtable}
    \hfill
    \begin{subtable}[t]{0.15\textwidth}
        \flushright
        \caption*{Wrote($\tid$,~aid,~pid)}
        \begin{tabular}{| c ? c | c |}
            \thinhline
            $\tid$&aid&pid\\
            \thickhline
            $t_{14}$&$a_1$&$p_1$\\
            \thinhline
            $t_{15}$&$a_2$&$p_1$\\
            \thinhline
            $t_{16}$&$a_3$&$p_2$\\
            \thinhline
            $t_{17}$&$a_6$&$p_3$\\
            \thinhline
            $t_{18}$&$a_7$&$p_3$\\
            \thinhline
            $t_{19}$&$a_7$&$p_4$\\
            \thinhline
            $t_{20}$&$a_8$&$p_4$\\
            \thinhline
            $t_{21}$&$a_6$&$p_5$\\
            \thinhline
        \end{tabular}
    \end{subtable}
    \hfill
    \begin{subtable}[t]{\textwidth}
        \begin{align*}
            & \delta_1=\Author(t,a,n,i) \wedge \Author(t',a,n',i') \wedge n \neq n' \rightarrow \bot; \qquad \hro_1=\Author(t,x,n,i) \wedge \Author(t',y,n,i) \harrow \linkO(x,y)\\
            & \delta_2=\Paper(t,p,\mathit{ti},c,a) \wedge \Wrote(t',a,p) \rightarrow \bot; \; \; \, \hrv_1=\Author(x,a,n,i) \wedge \Author(y,a,n',i') \wedge n\approx n' \harrow \linkV(\tup{x,2},\tup{y,2})
        \end{align*}
        \begin{center}
            $\sro_1=\Author(t,x,n,i) \wedge \Author(t',y,n',i') \wedge n \approx n' \wedge \Wrote(t'',x,p) \wedge \Wrote(t''',y,p) \sarrow \linkO(x,y)$
        \end{center}
    \end{subtable}
    \caption{Schema $\S_\ex$, $\S_\ex$-database $D_\ex$, and ER specification $\E_\ex=\langle \Gamma^{O}_\ex, \Gamma^{V}_\ex, \Delta_\ex \rangle$ with $\Gamma^O_\ex=\{\hro_1,\sro_1\}$, $\Gamma^{V}_\ex=\{\hrv_1\}$, and $\Delta_\ex=\{\delta_1,\delta_2\}$. 
    }\label{fig:mainEx}
\end{figure*}

\subsection{Semantics of \textsc{Lace}$^{\tiny \text{+}}$ Specifications}
In a nutshell, the semantics is based upon 
considering sequences of rule applications that result in a database
that satisfies the hard rule and denial constraints. 
Every such sequence 
gives rise to a solution, which takes the form of a pair of 
equivalence relations $\tup{\merO,\merV}$, specifying which objects and cells have been merged.  
Importantly, rules and constraints are evaluated w.r.t.\ the induced database, taking 
into account previously derived merges of objects and cells. 

In the original \textsc{Lace} framework, solutions consist of a single equivalence relation over objects, 
and induced databases are simply defined as the result of replacing every object with a 
representative of its equivalence class. Such an approach cannot however accommodate
local identifications of values. 
For this reason, we shall work with an extended form of database, where the arguments 
are  
\emph{sets of constants}. 

\begin{definition}
Given an $\S$-database $D$, equivalence relation $\merO$ over $\ObjD(D)$, and equivalence relation $\merV$ over $\StrPos(D)$, 
we denote by $\extdat{D}{\merO}{\merV}$ the \emph{(extended) database induced by $D$, $E$, and $V$}, which is obtained from $D$ by replacing:
\begin{itemize}
\item each tid $t$ with the singleton set $\{t\}$,
\item each occurrence of $o \in \ObjD(D)$ by $\{o' \mid (o, o') \in \merO\}$,
\item each value in a cell $\tup{t,i} \in \StrPos(D)$ with the set of values $\{t'[i'] \mid (\tup{t,i}, \tup{t',i'}) \in \merV \}$. 
\end{itemize}
\end{definition}
%

It remains to specify how queries  in rule bodies and constraints are to be 
evaluated over such induced databases. 
First, we need to say how similarity predicates 
are extended to sets of constants. 
We propose that $C_1 \approx C_2$ is satisfied whenever there are $c_1 \in C_1$ and $c_2 \in C_2$ such that $c_1 \approx c_2$,
since the elements of a set provide different possible representations of a value. 
Second, we must take care when handling join variables in value positions. Requiring all occurrences of a variable to 
map to the same set is too strong, e.g.\  it forbids us from matching \{J.\ Smith, Joe Smith\}
with \{J.\ Smith\}. We require instead that the intersection of all sets of constants assigned to a given variable
is non-empty. 

\begin{definition}\label{queryeval}
A Boolean query $q$ (possibly containing similarity and inequality atoms) is 
\emph{satisfied in} 
$\extdat{D}{\merO}{\merV}$, denoted $\extdat{D}{\merO}{\merV} \models q$, if there exists a function $h: \vars(q) \cup \qcons(q) \rightarrow 2^{\dcons(D)} \setminus \{\emptyset\}$ and 
functions $g_\ratom: \{0, \ldots, k\} \rightarrow 2^{\dcons(D)}$ for each $k$-ary relational atom $\ratom\in q$, such that: 
\begin{enumerate}
\item $h$ is determined by the $g_\pi$:  for every $a \in \qcons(q)$, $h(a)=\{a\}$, and for every $z \in \vars(q)$, $h(z)$ is the intersection of all sets $g_\ratom(i)$ such that $z$ is the $i$th argument of $\ratom$;
\item for every relational atom $\ratom = R(u_0, u_1, \ldots, u_k) \in q$, $R(g_\ratom(0), g_\ratom(1), \ldots, g_\ratom(k)) \in \extdat{D}{\merO}{\merV}$,
and for every $1 \leq i \leq k$, if $u_i \in \qcons(q)$, then $u_i \in g_\ratom(i)$;
\item for every inequality atom $z \neq z' \in q$: $h(z) \cap h(z') = \emptyset$;
\item for every similarity atom $u \approx u' \in q$: there exist $c \in h(u)$ and $c' \in h(u')$ such that $c \approx c'$.
\end{enumerate}
For non-Boolean queries, the set $q(\extdat{D}{\merO}{\merV}\!)$ of answers to $q(\vec{x})$
contains those tuples $\vec{c}$ such that $\extdat{D}{\merO}{\merV} \models q[\vec{c}]$. 
%
%
\end{definition}
Observe that the functions $g_\ratom$ make it possible to map the same variable $z$ to different sets, with Point 1 ensuring these sets have a non-empty intersection, $h(z)$. It is this intersected set, storing the common values for $z$, that is used to evaluate inequality and similarity atoms. 
Note that when constants occur in relational atoms, the sets assigned to a constant's position must contain that constant. 

The preceding definition of satisfaction of queries is straightforwardly extended to constraints and rules: 
\begin{itemize}
\item 
$\extdat{D}{\merO}{\merV} \models \exists \vec{y} \per \varphi(\vec{y}) \rightarrow \bot$ iff $\extdat{D}{\merO}{\merV} \not \models \exists \vec{y} \per \varphi(\vec{y})$
\item $\extdat{D}{\merO}{\merV} \models q(x,y) \rightarrow \linkO(x,y)$ iff $q(\extdat{D}{\merO}{\merV}) \subseteq \merO$
\item $\extdat{D}{\merO}{\merV} \models q(x_t,y_t) \rightarrow \linkV(\tup{x_t,i},\tup{y_t, j})$ iff
$(t_1,t_2) \in q(\extdat{D}{\merO}{\merV})$ implies $(\tup{t_1,i},\tup{t_2,j}) \in V$;
\end{itemize}
where symbol $\rightarrow$ can be instantiated by either $\harrow$ or $\sarrow$. We write $\extdat{D}{\merO}{\merV} \models \Lambda$ iff $\extdat{D}{\merO}{\merV} \models \lambda$ for every $\lambda \in \Lambda$.



With these notions in hand, we can formally define solutions of \textsc{Lace}$^{\tiny \text{+}}$ specifications. 

\begin{definition}\label{def:SOL}
    Given an 
    ER specification $\E=\tup{\hsrulesO,\hsrulesV,\Delta}$ over schema $\S$ and an $\S$-database $D$, we call 
    $\tup{\merO,\merV}$ a \emph{candidate solution for $(D,\E)$} if it satisfies one of the following: 
    \begin{itemize}
        \item $\merO=\eqrel(\emptyset,\ObjD(D))$ and $\merV=\eqrel(\emptyset,\StrPos(D))$; 
        \item $\merO=\eqrel(\merO' \cup \{(o,o')\},\ObjD(D))$,  where 
        $\tup{\merO',\merV}$ is a candidate solution
        for $(D,\E)$ and $(o,o')  \in q(\extdat{D}{\merO}{\merV})$ for some $q(x,y) \rightarrow \linkO(x,y) \in \hsrulesO$; 
        \item $V=\eqrel(\merV' \cup \{(\tup{t,i},\tup{t',i'})\},\StrPos(D))$, where
        $\tup{\merO,\merV'}$ is a candidate solution for $(D,\E)$ and 
        $(t,t') \in q(\extdat{D}{\merO}{\merV}\!)$ for some $q(x_t,y_t) \!\rightarrow \!\linkV(\!\tup{x_t,i},\tup{y_t,i'}\!) \in \hsrulesV$.
    \end{itemize}
If also 
$\extdat{D}{\merO}{\merV} \models \HRO \cup \HRV \cup \Delta$, 
then 
$\tup{\merO,\merV}$ 
is a \emph{solution} for $(D,\E)$. We use $\SOL(D,\E)$ for the set of solutions for $(D,\E)$.
\end{definition}

We return to our running example to illustrate solutions and the utility of local merges: 

\begin{example}\label{ex:SOL} 
Starting from database $D_{\ex}$, we can apply the soft rule 
$\sro_1$ to merge author ids $a_1$ and $a_2$ (more formally, we minimally 
extend the initial trivial equivalence relation $E$ to include $(a_1,a_2)$). 
The resulting induced instance is obtained by replacing \emph{all} occurrences of 
$a_1$ and $a_2$ by $\{a_1,a_2\}$. 
Note that the constraint $\delta_1$ is now violated, since $t_1$ and $t_2$
match on aid, but have different names. In the original \textsc{Lace} framework, 
this would prevent $(a_1,a_2)$ from belonging to any solution. 
However, thanks to the hard rule for values $\hrv_1$, we can resolve this violation. 
Indeed, $\hrv_1$  is applicable and allows us to (locally) merge 
the names in facts $t_1$ and $t_2$. 
The new induced database
contains \{J.\ Smith, Joe Smith\} in the name position of $t_1$ and $t_2$,
but the names for $t_3$, $t_4$, $t_5$ remain as before. 
Note the importance of performing a local rather than a global merge: if we had grouped
 J. Smith with Joe Smith \emph{everywhere}, this would force a merge of $a_3$ with $a_4$ due to the hard rule $\hro_1$, 
which would in turn violate $\delta_2$, again resulting in no solution containing $(a_1,a_2)$.  
Following the local merge of the names of $t_1$ and $t_2$, 
the hard rule $\hro_1$ becomes applicable and allows us (actually, forces us) 
to merge (globally) author ids $a_1$ and $a_5$. 
We let $\tup{\merO_{\ex},\merV_{\ex}}$ be the equivalence relations obtained from 
the preceding 
rule applications. As the instance induced by $\tup{\merO_{\ex},\merV_{\ex}}$
satisfies all hard rules and constraints, $\tup{\merO_{\ex},\merV_{\ex}}$ is a solution. 
Another solution is 
the pair of trivial equivalence relations,
since $D_{\ex}$ satisfies the constraints and hard rules. 
\end{example}
%

Similarly to~\cite{BienvenuPODS22}, we will compare solutions 
w.r.t.\ set inclusion, to maximize the discovered merges.  

\begin{definition}\label{def:maxsol}
    A solution $\tup{\merO,\merV}$ for $(D,\E)$ is a \emph{maximal solution} for $(D,\E)$ if there exists no solution $\tup{\merO',\merV'}$ for $(D,\E)$ such that 
    $\merO \cup \merV \subsetneq \merO' \cup \merV'$. 
%
    We denote by $\MaxSOL(D,\E)$ the set of maximal solutions for $(D,\E)$.
\end{definition}

\begin{example}
    The solution $\tup{\merO_{\ex},\merV_{\ex}}$ described in Example~\ref{ex:SOL} is not optimal 
    as the soft rule $\sro_1$ can be applied to get 
    $(a_6,a_7)$ or  $(a_7,a_8)$. 
    Notice, however, that it is not possible to include both merges, otherwise by transitivity, $a_6, a_7, a_8$ would all be replaced by $\{a_6,a_7,a_8\}$, 
    which would violate denial $\delta_1$ due to paper $p_5$. We have two maximal solutions: a first that extends $\tup{\merO_{\ex},\merV_{\ex}}$ with $(a_6,a_7)$ and the corresponding pair of names cells $(\tup{t_6,2},\tup{t_7,2})$ (due to $\hrv_1$),  
    and a second that extends $\tup{\merO_{\ex},\merV_{\ex}}$ with $(a_7,a_8)$ and the corresponding name cells $(\tup{t_6,2},\tup{t_7,2})$ (again due to $\hrv_1$). 
\end{example}

The \textsc{Lace}$^{\tiny \text{+}}$ framework properly generalizes the one in 
\cite{BienvenuPODS22}: if we take $\E=\tup{\hsrulesO,\emptyset,\Delta}$ (i.e.\ no rules for values), 
then $\merO$ is a solution for $(D,\E)$ in the original \textsc{Lace}  framework iff 
$\tup{\merO \cap (\Obj \times \Obj),\eqrel(\emptyset,\StrPos(D))} \in \SOL(D,\E)$. 


%
%

More interestingly, we show that it is in fact possible to simulate global merges using local merges. 

\begin{restatable}{theorem}{SimulationThm}\label{simulation}
For every 
ER specification $\E=\tup{\hsrulesO,\hsrulesV,\Delta}$ over $\S$, 
there exists a 
specification $\E' = \tup{\emptyset,\hsrulesV',\Delta}$ (over a modified $\S$, with all object positions changed to value positions, and all object constants treated as value constants) such that for every $\S$-database $D$: 
$\SOL(D,\E') = \{\tup{\emptyset, \merV \cup \merV_\merO}  \mid \tup{\merO,\merV} \in \SOL(D,\E)\} $,
where $V_E$ contains all pairs $(\tup{t,i},\tup{t', j})$ such that 
$(t[i],t'[j]) \in E$. 
%
%
\end{restatable}




\section{Computational Aspects}\label{sec:Computational}
We briefly explore the computational properties of \textsc{Lace}$^{\tiny \text{+}}$. 
%
As in~\cite{BienvenuPODS22}, we are interested in the \emph{data complexity} of the following decision problems: \recpb (\resp \maxrecpb) which checks if $\tup{E,V} \in \SOL(D,\E)$ (\resp $\tup{E,V} \in \MaxSOL(D,\E)$), \existpb which determines if $\SOL(D,\E) \neq \emptyset$, \cmergepb (\resp \pmergepb) which checks if
a candidate merge belongs to $E \cup V$ 
 for 
all (\resp some) $\tup{E,V} \in \MaxSOL(D,\E)$, and \certanspb (\resp \possanspb) which checks whether $\vec{c} \in q(\extdat{D}{E}{V})$ for all (\resp some) $\tup{E,V} \in \MaxSOL(D,\E)$. Interestingly, we show that incorporating 
local merges does not affect the complexity of all the above decision problems. 

\begin{restatable}{theorem}{Complexity}\label{th:Complexity}
    \recpb is $\PTIME$-complete; \maxrecpb is $\coNP$-complete; \existpb, \pmergepb, and \possanspb are $\NP$-complete; \cmergepb and \certanspb are $\PItwop$-complete. For specifications that do not use inequality atoms in denial constraints, \recpb, \maxrecpb, and \existpb are $\PTIME$-complete; \pmergepb and \possanspb are $\NP$-complete; \cmergepb and \certanspb are $\coNP$-complete.
\end{restatable}

Due to Theorem \ref{simulation}, all lower bounds hold even for specifications that do not contain any rules for objects.

\newcommand{\PSol}{\ensuremath{\Pi_{\textit{Sol}}}\xspace}
\newcommand{\match}{\textit{match}\xspace}
\newcommand{\Sim}{\textit{Sim}\xspace}
\newcommand{\val}{\textit{Val}\xspace}
\newcommand{\cell}{\textsf{c}\xspace}
\newcommand{\EqV}{\textit{EqV}\xspace}
\newcommand{\EqO}{\textit{EqO}\xspace}
\newcommand{\occ}{\textsf{vpos}}
\newcommand{\Oocc}{\textsf{opos}}
\newcommand{\proj}{\textit{Proj}}
\newcommand{\obj}{\textit{Obj}}
An ASP encoding of the original \textsc{Lace}  framework was proposed in \cite{BienvenuPODS22}. We can extend that 
encoding to obtain a normal logic program \PSol whose stable models capture  \textsc{Lace}$^{\tiny \text{+}}$  solutions: 
\begin{restatable}{theorem}{thmASPSols}\label{thm:asp-sols}
For every database $D$ and 
specification $\E=\tup{\hsrulesO,\hsrulesV,\Delta}$: 
$\tup{\merO,\merV}\in \SOL(D,\E)$ iff $\merO = \{(a,b) \mid \EqO(a,b) \in M\}$ and $\merV = \{(\tup{t,i},\tup{t',i'}) \mid \EqV(t,i,t',i') \in M\}$
for a stable model $M$ of $(\Pi_{Sol}, D)$. 
\end{restatable}

We refer to the appendix for the details and sketch here how rules for values are handled. 
Basically, every hard rule 
$q(x_t,y_t) \harrow  \linkV(\tup{x_t,i},\tup{y_t,j}) $ 
is translated into the ASP rule $\EqV(x_t,i,y_t,j) \leftarrow \hat{q}(x_t,y_t)$.
To define $\hat{q}$, we use $\occ(v)$ (resp.\ \Oocc) for 
the set of pairs $(u_t,i)$ 
such that $v$ occurs in a value (resp.\ object) position $i$ in atom $R(u_t,v_1, \dots, v_k) \in q$.
The query $\hat{q}$ is obtained from $q$ by replacing each occurrence $(u_t,i)$ of a 
non-distinguished variable $v$ in $q$ 
with a fresh variable $v_{(u_t,i)}$, and then: 

\begin{itemize}
\item  for every join variable $v$ in $q$, take fresh variables $u'_t,k,v'$ and add to
$\hat{q}$ the set of atoms $\{EqV(u_t, i, u'_t, k) \mid (u_t,i) \in \occ(v) \} \cup \{\EqO(v_{(u_t,i)},v') \mid (u_t,i) \in \Oocc(v) \}$; 
 \item for each 
atom 
$\alpha= v \approx w$, take fresh variables $v',w'$ and replace $\alpha$  by the set of atoms
 $\{\val(u_t, i, v' )| (u_t,i) \in \occ(v) \} \cup \{\val(u'_t, j, w' )| (u'_t,j) \in \occ(w) \} \cup \{v' \approx w'\}$, where $\val$
 is a predicate defined by the rule:
 $$\val(u_t,i,v) \leftarrow \EqV(u_t,i,u'_t,j), \proj(u'_t,j,v), \quad \text{and}$$
 ground atoms $\proj(t,i,c)$ of $\proj/3$ encode $t[i]= c$.
%
\end{itemize}
Soft rules for values are handled similarly: we use the same modified body $\hat{q}$,
but then enable a choice between producing $\EqV(x_t,i,y_t,j)$ or not applying the rule (adding a blocking fact $\textit{NEqV}(x_t,i,y_t,j)$).
Additionally, \PSol will contain rules that encode object rules (producing $\EqO$ facts), rules that ensure $\EqV$ and 
$\EqO$ are equivalence relations, 
and rules that enforce the satisfaction of the denial constraints.

\section*{Acknowledgements}
This work has been supported by the ANR AI Chair INTENDED (ANR-19-CHIA-0014), by MUR under the PNRR project FAIR (PE0000013), and by the Royal Society (IES\textbackslash R3\textbackslash 193236).

\bibliographystyle{kr}
\bibliography{refs}

\begin{thebibliography}{}

\bibitem[\protect\citeauthoryear{Arasu, R{\'{e}}, and
  Suciu}{2009}]{ArasuICDE09}
Arasu, A.; R{\'{e}}, C.; and Suciu, D.
\newblock 2009.
\newblock Large-scale deduplication with constraints using dedupalog.
\newblock In {\em Proceedings of the Twenty-Fifth International Conference on
  Data Engineering (ICDE~2009)},  952--963.

\bibitem[\protect\citeauthoryear{Benjelloun \bgroup et al\mbox.\egroup
  }{2009}]{BenjellounVLDBJ09}
Benjelloun, O.; Garcia{-}Molina, H.; Menestrina, D.; Su, Q.; Whang, S.~E.; and
  Widom, J.
\newblock 2009.
\newblock Swoosh: a generic approach to entity resolution.
\newblock {\em The VLDB Journal} 18(1):255--276.

\bibitem[\protect\citeauthoryear{Bertossi, Kolahi, and
  Lakshmanan}{2013}]{BertossiTCS13}
Bertossi, L.~E.; Kolahi, S.; and Lakshmanan, L. V.~S.
\newblock 2013.
\newblock Data cleaning and query answering with matching dependencies and
  matching functions.
\newblock {\em Theory of Computing Systems} 52(3):441--482.

\bibitem[\protect\citeauthoryear{Bertossi}{2011}]{BertossiMC2011}
Bertossi, L.~E.
\newblock 2011.
\newblock {\em Database Repairing and Consistent Query Answering}.
\newblock Synthesis Lectures on Data Management. Morgan {\&} Claypool
  Publishers.

\bibitem[\protect\citeauthoryear{Bhattacharya and
  Getoor}{2007}]{BhattacharyaTKDD07}
Bhattacharya, I., and Getoor, L.
\newblock 2007.
\newblock Collective entity resolution in relational data.
\newblock {\em ACM Transactions on Knowledge Discovery from Data} 1(1):5.

\bibitem[\protect\citeauthoryear{Bienvenu, Cima, and
  Guti{\'{e}}rrez{-}Basulto}{2022}]{BienvenuPODS22}
Bienvenu, M.; Cima, G.; and Guti{\'{e}}rrez{-}Basulto, V.
\newblock 2022.
\newblock {LACE:} {A} logical approach to collective entity resolution.
\newblock In {\em Proceedings of the Forty-First {ACM}
  {SIGMOD}-{SIGACT}-{SIGAI} Symposium on Principles of Database Systems
  (PODS~2022)},  379--391.

\bibitem[\protect\citeauthoryear{Bienvenu, Cima, and
  Guti{\'{e}}rrez{-}Basulto}{2023}]{BienvenuIJCAI23}
Bienvenu, M.; Cima, G.; and Guti{\'{e}}rrez{-}Basulto, V.
\newblock 2023.
\newblock {REPLACE}: {A} logical framework for combining collective entity
  resolution and repairing.
\newblock In {\em Proceedings of IJCAI 2023}.

\bibitem[\protect\citeauthoryear{Burdick \bgroup et al\mbox.\egroup
  }{2016}]{BurdickTODS2016}
Burdick, D.; Fagin, R.; Kolaitis, P.~G.; Popa, L.; and Tan, W.
\newblock 2016.
\newblock A declarative framework for linking entities.
\newblock {\em ACM Transactions on Database Systems} 41(3):17:1--17:38.

\bibitem[\protect\citeauthoryear{Christophides \bgroup et al\mbox.\egroup
  }{2021}]{ChristophidesACMSurvey21}
Christophides, V.; Efthymiou, V.; Palpanas, T.; Papadakis, G.; and Stefanidis,
  K.
\newblock 2021.
\newblock An overview of end-to-end entity resolution for big data.
\newblock {\em ACM Computing Surveys} 53(6):127:1--127:42.

\bibitem[\protect\citeauthoryear{Deng \bgroup et al\mbox.\egroup
  }{2022}]{DengICDE22}
Deng, T.; Fan, W.; Lu, P.; Luo, X.; Zhu, X.; and An, W.
\newblock 2022.
\newblock Deep and collective entity resolution in parallel.
\newblock In {\em Proceedings of the Thirty-Eighth {IEEE} International
  Conference on Data Engineering (ICDE~2022)},  2060--2072.

\bibitem[\protect\citeauthoryear{Fan and Geerts}{2012}]{FanMC2012}
Fan, W., and Geerts, F.
\newblock 2012.
\newblock {\em Foundations of Data Quality Management}.
\newblock Morgan {\&} Claypool Publishers.

\bibitem[\protect\citeauthoryear{Fan \bgroup et al\mbox.\egroup
  }{2009}]{FanVLDB09}
Fan, W.; Jia, X.; Li, J.; and Ma, S.
\newblock 2009.
\newblock Reasoning about record matching rules.
\newblock {\em Proceedings of the VLDB Endowment} 2(1):407--418.

\bibitem[\protect\citeauthoryear{Fan}{2008}]{FanPODS08}
Fan, W.
\newblock 2008.
\newblock Dependencies revisited for improving data quality.
\newblock In {\em Proceedings of the Twenty-Seventh {ACM}
  {SIGMOD-SIGACT-SIGART} Symposium on Principles of Database Systems
  (PODS~2008)},  159--170.
\newblock {ACM}.

\bibitem[\protect\citeauthoryear{Singla and Domingos}{2006}]{SinglaICDM06}
Singla, P., and Domingos, P.~M.
\newblock 2006.
\newblock Entity resolution with markov logic.
\newblock In {\em Proceedings of the Sixth {IEEE} International Conference on
  Data Mining (ICDM~2006)},  572--582.

\end{thebibliography}

\appendix
\onecolumn
\def\dval{D^{\Val}}

Before giving the proof of Theorem 1, we start by being more precise what we mean by applicability and applications of rules. 
We say that a rule $\rho=q(x,y) \rightarrow \linkO(x,y) \in \hsrulesO$
is applicable in $\extdat{D}{\merO}{\merV}$ if there exists $(o,o')  \in q(\extdat{D}{\merO}{\merV})$ 
such that $(o,o') \not \in \merO$. Applying this rule will means adding $(o,o')$ to $\merO$ (then closing under symmetry and transitivity to get back to an equivalence relation, as implemented by $\eqrel$). 
Similarly, we say that a rule 
$q(x_t,y_t) \rightarrow \linkV(\tup{x_t,i},\tup{y_t,i'}) \in \hsrulesV$ is applicable in 
$\extdat{D}{\merO}{\merV}$ if there exists  $(t,t') \in q(\extdat{D}{\merO}{\merV})$ 
such that $(\tup{t,i},\tup{t',i'}) \not \in \merV$. In this case, applying the rule means adding 
$(\tup{x_t,i},\tup{y_t,i'})$ to $\merV$ (and then closing the set under symmetry and transitivity). 

\SimulationThm*

\begin{proof}
Take a \textsc{Lace}$^{\tiny \text{+}}$ ER specification $\E=\tup{\hsrulesO,\hsrulesV,\Delta}$ over $\S$. 
We define the new set $\hsrulesV'$ 
as $\hsrulesV \cup \Gamma_{\Obj \leadsto\Val} \cup \Gamma_{\mathsf{global}}$, where:
\begin{itemize}
\item $\Gamma_{\Obj \leadsto\Val}$ contains, for every rule for objects $\rho=q(x,y) \rightarrow \linkO(x,y) \in \hsrulesO$, the following rule for values
$$\rho^v= q(x_t,y_t) \rightarrow \linkV(\tup{x_t,i},\tup{y_t, j})$$
where $q(x_t,y_t)$ has the same atoms as $q(x,y)$ (simply $x_t$ and $y_t$ appear in the rule head, rather than $x,y$), and $\tup{x_t,i}$ and $\tup{y_t, j}$ are (any arbitrarily selected) positions in $q$ that contain $x$ and $y$, respectively. 
\item $\Gamma_{\mathsf{global}}$ contains all hard rules for values of the form 
$$P(x_t,u_1, \ldots, u_{i-1}, z, u_{i+1}, \ldots, u_k) \wedge P'(y_t,v_1, \ldots, v_{j-1}, z, v_{j+1}, \ldots, v_\ell) \harrow \linkV(\tup{x_t,i},\tup{y_t, j})$$
where $P/k,P'/\ell \in \S$ and $\typef(P,i) = \typef(P',j)= \Obj$ (w.r.t.\ the original schema $\S$). 
\end{itemize}
Intuitively, every rule for objects $\rho$ is replaced by a rule $\rho^v$ for values that will merge a single pair of cells containing the variables $x,y$. 
The rules in $\Gamma_{\mathsf{global}}$ ensure that all cells that contain the same object constant (w.r.t.\ the original schema) will be in the same equivalence class, hence merging a single pair of cells with object constants ensures that all other cells containing occurrences of these objects (now treated as values) will belong to the same equivalence class. 

Note that when computing $\SOL(D,\E')$, all objects in $D$ will be treated as elements of $\Val$. To keep track of the two different settings in the proof, we shall use $\S_{\Val}$ for the schema obtained from $\S$ by changing all object positions to value positions, and we shall use $D^{\Val}$ to indicate that we are considering the database $D$ w.r.t.\ this modified schema, i.e., where all constants $o \in \ObjD(D)$ are now assumed to belong instead to $\Val$. Observe that, $\StrPos(D^{\Val})$ now additionally contains all pairs $\tup{t,i}$ such that $t[i]\in \ObjD(D)$. We use the notation $\StrPos_\Obj(D^{\Val})$ 
for the subset of 
$\StrPos(D^{\Val})$ containing pairs $\tup{t,i}$ such that $t[i] \in \Obj$ (w.r.t.\ the original schema), 
and will call $\tup{t,i}$ an \emph{object cell}.
Likewise, we use $\StrPos_\Val(D^{\Val})$
for the subset of 
$\StrPos(D^{\Val})$ containing pairs $\tup{t,i}$ such that $t[i] \in \Val$ (w.r.t.\ the original schema).

To complete the proof, we must show that for every $\S$-database $D$: $$\SOL(D^{\Val},\E') = \{\tup{\emptyset, \merV \cup \merV_\merO}  \mid \tup{\merO,\merV} \in \SOL(D,\E)\}, $$
where $V_E$ contains all pairs $(\tup{t,i},\tup{t', j})$ such that $(t[i],t'[j]) \in E$. Note that we use $D^{\Val}$ in $\SOL(D^{\Val},\E')$ to emphasize that we are using the version of $D$ with objects treated as values. \medskip

\noindent ($\Leftarrow$).  
First we suppose that $\tup{\merO,\merV} \in \SOL(D,\E)$ and show that $\tup{\emptyset, \merV \cup \merV_\merO}$ belongs to 
$\SOL(D^{\Val},\E')$.  
By Definition~\ref{def:SOL}, there exists a sequence $\tup{\merO_0,\merV_{0}}, \tup{\merO_1,\merV_{1}},  \ldots, \tup{\merO_n,\merV_n}$
such that $\merO_0=\eqrel(\emptyset,\ObjD(D))$, $\merV_0=\eqrel(\emptyset,\StrPos(D))$, 
$\tup{\merO,\merV}= \tup{\merO_n,\merV_n}$, 
 and for every $0 \leq j < n$, one of the following holds:
 \begin{itemize}
     \item $\tup{\merO_{j+1},\merV_{j+1}}= \tup{\eqrel(\merO_j \cup \{(o_j,o_j')\},\ObjD(D)),\merV_{j}}$ where 
        $(o_j,o_j')  \in q(\extdat{D}{\merO_j}{\merV_j})$ for $\rho_j=q(x,y) \rightarrow \linkO(x,y) \in \hsrulesO$, 
        \item $\tup{\merO_{j+1},\merV_{j+1}} = \tup{\merO_{j}, \eqrel(\merV_j \cup \{(\tup{t_j,i_j},\tup{t_j',i_j'})\},\StrPos(D)}$
        where 
        $(t_j,t_j') \in q(\extdat{D}{\merO_j}{\merV_j})$ for some $\rho_j = q(x_t,y_t) \rightarrow \linkV(\tup{x_t,i_j},\tup{y_t,i'_j}) \in \hsrulesV$.
    \end{itemize}
Moreover, we know that $\extdat{D}{\merO}{\merV} \models \HRO \cup \HRV \cup \Delta$. 

We can now build a corresponding sequence of candidate solutions for $(D^{\Val},\E')$
that witnesses that $\tup{\emptyset, \merV \cup \merV_\merO} \in \SOL(D^{\Val},\E')$. 
The sequence will have the form 
$$\tup{\emptyset,\merV^{\mathsf{init}}_{0}}, \ldots, \tup{\emptyset,\merV^{\mathsf{init}}_{w}},  \tup{\emptyset,\merV^*_{0}}, \tup{\emptyset,\merV^*_{1}},  \ldots, \tup{\merO_n,\merV^*_n}$$
where we first exhaustively apply the rules in $\Gamma_{\mathsf{global}}$ to $\dval$, giving us the initial sequence 
$\tup{\emptyset,\merV^{\mathsf{init}}_{0}}, \ldots, \tup{\emptyset,\merV^{\mathsf{init}}_{w}},  \tup{\emptyset,\merV^*_{0}}$. After  having completed all possible rule applications of $\Gamma_{\mathsf{global}}$, we will have the equivalence relation $\merV^*_{0}$ on $\StrPos(D^{\Val})$, whose 
equivalence classes will be of two types: (i) singleton sets $\{\tup{t,i}\}$ containing a value cell,
or (ii) sets containing all object cells that contain a given constant $c \in \ObjD(D)$. At this point, we can essentially reproduce the original sequence of rule applications, using the rules in $\hsrulesV \cup \Gamma_{\Obj \leadsto\Val}$, to produce the remainder of the sequence $\tup{\emptyset,\merV^*_{0}}, \tup{\emptyset,\merV^*_{1}},  \ldots, \tup{\merO_n,\merV^*_n}$. When $\rho_j \in \hsrulesV$, we can keep the same rule application which adds the pair of value cells $(\tup{t_j,i_j},\tup{t_j',i_j'})$.
However, when $\rho_j \in \hsrulesO$, we must use the corresponding rule for values $\rho^v$. Instead of adding $(o_j,o_j')$ to $\merO_j$, the application of $\rho^v$ will add to $\merV^*_{j}$ a single pair of object cells $(\tup{t_j,i_j},\tup{t_j',i_j'})$  such that 
$t_j[i_j]=o_j$ and $t_j'[i_j']=o_j'$. It can be shown by induction that by proceeding in this manner, we have the following properties, for all $0 \leq j \leq h$:
 \begin{itemize}
 \item $\merV^*_{j} = \merV_{j} \cup \{(\tup{\hat{t},\ell},\tup{\hat{t}', \ell'}) \mid (\hat{t}[\ell], \hat{t}'[\ell']) \in E_j \}$
 \item $D_{E_j,V_j} = \dval_{\emptyset,V^*_j}$ (modulo the fact that object constants in $D_{E_j,V_j}$ are considered as value constants in $\dval_{\emptyset,V^*_j}$)
 \end{itemize}
In particular, this means that we have $\tup{\emptyset,\merV^*_n}= \tup{\emptyset, \merV \cup \merV_\merO}$.
Moreover, since $\extdat{D}{\merO}{\merV} \models \HRO \cup \HRV \cup \Delta$, we immediately get $\dval_{\emptyset,V^*_n} \models \HRV \cup \Delta$, and we can use  $\extdat{D}{\merO}{\merV} \models \HRO$, and the fact that $\rho$ and $\rho^v$
share the same rule body, to show that we also have $\dval_{\emptyset,V^*_n} \models \rho^v$ for all hard rules in $\Gamma_{\Obj \leadsto\Val}$. By construction, we also have $\dval_{\emptyset,V^*_n} \models \Gamma_{\mathsf{global}}$. 
It follows that $\tup{\emptyset, \merV \cup \merV_\merO} \in \SOL(D^{\Val},\E')$. \medskip

\noindent ($\Rightarrow$). For the other direction, take some $(\merO^*,\merV^*) \in \SOL(D^{\Val},\E')$. Then $\ObjD(D^{\Val})=\emptyset$, so we must have $\merO^* = \emptyset$.
Moreover, due to the structure of the new rules, for every pair $(\tup{t,i}, \tup{t',i'}) \in \merV^* \setminus (\StrPos_\Val(D^{\Val}) \times \StrPos_\Val(D^{\Val}))$, it must be the case that $t[i], t'[i'] \in \Obj$ (w.r.t.\ the original schema and definition of $\Obj$ and $\Val$), hence, $(\tup{t,i}, \tup{t',i'}) \in \merV^* \cap (\StrPos_\Obj(\dval) \times \StrPos_\Obj(\dval))$. We can thus partition $V^*$ into $V^*_\Val = V^* \cap (\StrPos_\Val(D^{\Val}) \times \StrPos_\Val(D^{\Val}))$ (containing pairs of value cells of $D$)
and $V^*_\Obj = V^* \cap (\StrPos_\Obj(D^{\Val}) \times \StrPos_\Obj(D^{\Val}))$ (containing pairs of object cells). 

By Definition \ref{def:SOL}, $\extdat{\dval}{\emptyset}{\merV^*} \models\Delta$ and $\extdat{\dval}{\emptyset}{\merV^*} \models\rho$ for every hard rule $\rho \in \hsrulesV'$. 
Furthermore, there exists a sequence $(\emptyset,\merV^*_0), (\emptyset,\merV^*_1),  \ldots, (\emptyset,\merV^*_m)$
of candidate solutions of $\SOL(D^{\Val},\E')$
such that:  
\begin{itemize}
\item $\merV^*_0=\eqrel(\emptyset,\StrPos(D^{\Val}))$, 
\item $\merV^*_m=\merV^*$, and
\item  for every $0 \leq k \leq m-1$, $\merV^*_{k+1}=\eqrel(\merV^*_k \cup \{(\tup{t_k,i_k},\tup{t'_k,i'_k})\},\StrPos(\dval))$, where
        $(t_k,t_k') \in q(\extdat{\dval}{\emptyset}{\merV^*_k}\!)$ \newline for some rule $\rho_k=q(x_t,y_t) \rightarrow \linkV(\tup{x_t,i_k},\tup{y_t,i'_k}) \in \hsrulesV'$.
\end{itemize}
We define $V^*_{\Obj,k} = V^*_k \cap V^*_\Obj $ and $V^*_{\Val,k} = V^*_k \cap V^*_\Val$.

We may \emph{assume w.l.o.g.\ that the hard rules in $\Gamma_{\mathsf{global}}$ are applied as soon as they are applicable}, since such rule applications must occur and moving them earlier cannot prevent the other rule applications from being performed later (this is due to the monotonicity of query answers w.r.t.\ extending the equivalence relations, cf.\ Lemma \ref{lem:monotone}). 

Let $p_0, p_1, \ldots, p_h$ be all of the indices $k$ (in increasing order) such that $\rho_k \in \hsrulesV \cup \Gamma_{\Obj \leadsto\Val}$.
Then due to our assumption on the order of rule applications, we can show by induction that we have the following property, for all for every $0 \leq j \leq h$:
\begin{description}
\item[($\star$)] $V^*_{\Obj,p_j}$ is an equivalence relation on $\StrPos_\Obj(D^{\Val})$ such that for every equivalence class $S_\ell$, 
there is a subset $O_\ell \subseteq \ObjD(D)$ such that $\tup{t,i} \in S_\ell$ iff $t[i] \in O_\ell$
\end{description}
and furthermore, that for all $0 < j \leq h$, we have $(\dagger)$ $p_{j+1}=p_j+1$. The latter property means that 
the rules in $\Gamma_{\mathsf{global}}$ are applied exhaustively at the beginning, and then never again. 

It follows from ($\star$) that with every $V^*_{\Obj,p_j}$, with $0 \leq j \leq h$, we have a corresponding equivalence relation $\merO_j$ on $\ObjD(D)$
such that $(\tup{t,i},\tup{t', i'}) \in V^*_{\Obj,p_j}$ iff $(t[i],t'[j]) \in E_j$. 
For every $0 \leq j \leq h$, we let $\merV_{j} = V^*_{\Val,p_j}$. We then build the sequence
$$(\merO_0,\merV_{0}), (\merO_1,\merV_{1}),  \ldots, (\merO_h,\merV_h)$$
Note that by construction, we have $\merO_0=\eqrel(\emptyset,\ObjD(D))$ and $\merV_0=\eqrel(\emptyset,\StrPos(D))$. 
It can be proven by induction that, for every $0 \leq j \leq h$, $D_{E_j, V_j} = \dval_{\emptyset,\merV^*_{p_j}}$ and 
one of the following holds:
\begin{itemize}
\item $\rho_{p_j} = q(x_t,y_t) \rightarrow \linkV(\tup{x_t,i_{p_j}},\tup{y_t,i_{p_j}'})  \in \hsrulesV$, $\merO_{j+1} = \merO_j$, and $\merV_{j+1} = \eqrel(\merV_{j} \cup \{(\tup{t_{p_j},i_{p_j}},\tup{t'_{p_j},i'_{p_j}})\},\StrPos(D))$, where  $(t_{p_j},t_{p_j}') \in q(D_{E_j, V_j})$;
\item $\rho_{p_j} \in \Gamma_{\Obj \leadsto\Val}$, $\merV_{j+1} = \merV_j$,  and $\merO_{j+1} = \eqrel(\merO_j \cup \{(o_j,o_j')\}, \ObjD(D))$, where 
$\rho_{p_j}= \tau^v =  q(x_t,y_t) \rightarrow \linkV(\tup{x_t,i},\tup{y_t, i'})$, $\tau=q(x,y) \rightarrow \linkO(x,y) \in \hsrulesO$, and 
$(o_j,o_j') \in q(D_{E_j, V_j})$
\end{itemize}
Moreover, since $D_{E_h, V_h} = \dval_{\emptyset,\merV^*_{p_h}} = \dval_{\emptyset,\merV^*}$ and $\extdat{\dval}{\emptyset}{\merV^*}$ satisfies $\Delta$ and all hard rules in $\hsrulesV'$, we immediately get $D_{E_h, V_h} \models \HRV \cup \Delta$. We can further use the fact that $\extdat{\dval}{\emptyset}{\merV^*}$ satisfies every hard rule $\rho^v \in \Gamma_{\Obj \leadsto\Val}$ to show that we also have $D_{E_h, V_h} \models \HRO$ (recall that $\rho^v \in \Gamma_{\Obj \leadsto\Val}$  has the same body as $\rho \in \hsrulesO$). It follows that $(E_h, V_h)$ is a solution to $(D, \E)$. 
To complete the proof, we observe that $(\merO^*,\merV^*)= (\emptyset,\merV^*)$ is such that 
$\merV^* = \merV_h \cup V_{E_h}$. 
\end{proof}

\noindent Before providing the proof of Theorem~\ref{th:Complexity}, we now introduce some intermediate results which are crucial to establish the upper bounds claimed in the theorem.

\begin{lemma}\label{lem:Eval}
    Let $q$ be a Boolean query (possibly involving similarity and inequality atoms) over a schema $\S$, $D$ be an $\S$-database, $E$ be an equivalence relation on $\ObjD(D)$, and $V$ be an equivalence relation on $\StrPos(D)$. Checking whether $\extdat{D}{E}{V} \models q$ can be done in polynomial time in the size of $D$, $E$, and $V$.
\end{lemma}

\begin{proof}
    First, it is immediate to verify that computing $\extdat{D}{E}{V}$, \ie the extended database induced by $D$, $E$, and $V$, can be done in polynomial time in the size of $D$, $E$, and $V$. With a slight abuse of notation, let $\dom{\extdat{D}{E}{V}}$ be the set of \emph{sets of constants} occurring in the extended database $\extdat{D}{E}{V}$. Assuming that $q$ is fixed, we now show that checking whether $\extdat{D}{E}{V} \models q$ can be done in polynomial time.
    
    Observe that, due to requirement 2.~of Definition~\ref{queryeval}, if there exists functions $h$ and $g_{\ratom}$ for each relational atom $\ratom \in q$ witnessing that $\extdat{D}{E}{V} \models q$, \ie satisfying requirements 1., 2., 3., and 4., then the image of each $g_{\ratom}$ must necessarily be a subset of $\dom{\extdat{D}{E}{V}}$. It follows that, when seeking for functions $h$ and $g_{\ratom}$ for each relational atom $\ratom \in q$ witnessing that $\extdat{D}{E}{V} \models q$, we can actually restrict the codomain of each $g_{\ratom}$ to $\dom{\extdat{D}{E}{V}}$ rather than considering $2^{\dcons(D)}$.

    Now, for each $k$-ary relational atom $\ratom \in q$, let $G_{\ratom}$ be the set of all possible functions $g_{\ratom}$ from $\{0, \ldots, k\}$ to $\dom{\extdat{D}{E}{V}}$. Since $q$ is assumed to be fixed, and therefore $k$ is fixed, the total number of such functions is only polynomial in the size of $D$, $E$, and $V$. Let $\ratom_1, \ldots, \ratom_n$ be the relational atoms occurring in $q$. From what said above, we have that $\extdat{D}{E}{V} \models q$ if and only if there exists functions $(g_{\ratom_1},\ldots,g_{\ratom_n})$ in the Cartesian product $G_{\ratom_1} \times \ldots \times G_{\ratom_n}$ such that, once computed the function $h$ as specified by requirement 1.~of Definition~\ref{queryeval}, requirements 2., 3., and 4.~are all satisfied. It is not hard to verify that \myi computing the function $h$ from $(g_{\ratom_1},\ldots,g_{\ratom_n})$ as specified in requirement 1.~and \myii checking whether requirements 2., 3., and 4.~are satisfied by $h$ and $(g_{\ratom_1},\ldots,g_{\ratom_n})$ can be both done in polynomial time. Since $q$ is assumed to be fixed, and therefore the cardinality of the Cartesian product $G_{\ratom_1} \times \ldots \times G_{\ratom_n}$ is only polynomial in the size of $D$, $E$, and $V$, we can immediately derive a polynomial time algorithm for checking whether $\extdat{D}{E}{V} \models q$ which just consider all the possible functions $(g_{\ratom_1},\ldots,g_{\ratom_n})$ in the Cartesian product $G_{\ratom_1} \times \ldots \times G_{\ratom_n}$. If for some $(g_{\ratom_1},\ldots,g_{\ratom_n}) \in G_{\ratom_1} \times \ldots \times G_{\ratom_n}$ requirements 2., 3., and 4.~are all satisfied after computing $h$ as specified in requirement 1., then we return \true (because $\extdat{D}{E}{V} \models q$); otherwise, we return \false (because $\extdat{D}{E}{V} \not\models q$). 
\end{proof}

\begin{lemma}\label{lem:monotone}
    Let $q(\vec{x})$ be an $n$-ary query (possibly involving similarity atoms) over a schema $\S$, $D$ be an $\S$-database, $\vec{c}$ be an $n$-tuple of constants, $E$ and $E'$ be two equivalence relations on $\ObjD(D)$, and $V$ and $V'$ be two equivalence relations on $\StrPos(D)$. If $\vec{c} \in q(\extdat{D}{E}{V})$ and $E \cup V \subseteq E' \cup V'$, then $\vec{c} \in q(\extdat{D}{E'}{V'})$.
\end{lemma}

\begin{proof}
    Suppose that $E \cup V \subseteq E' \cup V'$ and $\vec{c} \in q(\extdat{D}{E}{V})$, \ie $\extdat{D}{E}{V} \models q[\vec{c}]$. Since $E \cup V \subseteq E' \cup V'$, we can immediately derive the following: for each (extended) fact of the form $R(\{t\},C_1,\ldots,C_k)$ occurring in the extended database $\extdat{D}{E}{V}$, we have a corresponding (extended) fact of the form $R(\{t\},C'_1,\ldots,C'_k)$ occurring in the extended database $\extdat{D}{E'}{V'}$ such that $C_i \subseteq C'_i$ holds for every $1 \leq i \leq k$. Thus, since $q$ does not contain inequality atoms and since $\extdat{D}{E}{V} \models q[\vec{c}]$ holds by assumption, due to the previous observation and the notion of evaluation of Boolean queries over extended databases (Definition~\ref{queryeval}), it is not hard to see that $\extdat{D}{E'}{V'} \models q[\vec{c}]$ as well, \ie $\vec{c} \in q(\extdat{D}{E'}{V'})$.
\end{proof}

\begin{lemma}\label{lem:Preservation}
    Let $\E$ be an ER specification over a schema $\S$, $D$ be an $\S$-database, $q(\vec{x})$ be a CQ over $\S$, and $\vec{c}$ a tuple of constants. We have that $\vec{c} \in q(\extdat{D}{E}{V})$ for some $\tup{E,V} \in \MaxSOL(D,\E)$ if and only if $\vec{c} \in q(\extdat{D}{E'}{V'})$ for some $\tup{E',V'} \in \SOL(D,\E)$.
\end{lemma}

\begin{proof}
    First, suppose that $\vec{c} \in q(\extdat{D}{E}{V})$ holds for some $\tup{E,V} \in \MaxSOL(D,\E)$. Then, since $\MaxSOL(D,\E) \subseteq \SOL(D,\E)$, we have that $\vec{c} \in q(\extdat{D}{E}{V})$ for some $\tup{E,V} \in \SOL(D,\E)$ as well.

    Now, suppose that $\vec{c} \in q(\extdat{D}{E}{V})$, \ie $\extdat{D}{E}{V} \models q[\vec{c}]$, for some $\tup{E,V} \in \SOL(D,\E)$. Since $\tup{E,V} \in \SOL(D,\E)$, it follows that either $\tup{E,V} \in \MaxSOL(D,\E)$ or there exists at least one $\tup{E',V'}$ such that $E \cup V \subsetneq E' \cup V'$ and $\tup{E',V'} \in \MaxSOL(D,\E)$. In the former case, we trivially derive that $\vec{c} \in q(\extdat{D}{E}{V})$ for some $\tup{E,V} \in \MaxSOL(D,\E)$ and we are done. In the latter case, since $E \cup V \subsetneq E' \cup V'$, $\vec{c} \in q(\extdat{D}{E}{V})$, and $q(\vec{x})$ does not contain inequality atoms, by Lemma~\ref{lem:monotone} we derive that $\vec{c} \in q(\extdat{D}{E'}{V'})$ as well. So, we have that $\vec{c} \in q(\extdat{D}{E'}{V'})$ for some $\tup{E',V'} \in \MaxSOL(D,\E)$ also in this case, as required. 
\end{proof}

\begin{lemma}\label{lem:NoMoreMerg}
    Let $\delta$ be a denial constraint over a schema $\S$ that do not use inequality atoms, $D$ be an $\S$-database, $E$ and $E'$ be two equivalence relations on $\ObjD(D)$, and $V$ and $V'$ be two equivalence relations on $\StrPos(D)$. If $\extdat{D}{E}{V} \not\models \delta$ and $E \cup V \subseteq E' \cup V'$, then $\extdat{D}{E'}{V'}\not\models \delta$. 
\end{lemma}

\begin{proof}
    Let $\delta=\exists \vec{y} \per \varphi(\vec{y}) \rightarrow \bot$. By definition, we have that $\extdat{D}{E}{V} \not\models \delta$ iff $\extdat{D}{\merO}{\merV} \models \exists \vec{y} \per \varphi(\vec{y})$. So, suppose that $E \cup V \subseteq E' \cup V'$ and $\extdat{D}{E}{V} \not\models \delta$, \ie $\extdat{D}{E}{V} \models \exists \vec{y} \per \varphi(\vec{y})$. Since by assumption $\delta$ is a denial constraint without inequality atoms, and therefore $\exists \vec{y} \per \varphi(\vec{y})$ is a CQ (without inequality atoms), and since $\extdat{D}{E}{V} \models \exists \vec{y} \per \varphi(\vec{y})$ and $E \cup V \subseteq E' \cup V'$, by Lemma~\ref{lem:monotone} we derive that $\extdat{D}{E'}{V'} \models \exists \vec{y} \per \varphi(\vec{y})$ as well, and therefore $\extdat{D}{E'}{V'}\not\models \delta$, as required.
\end{proof}

\begin{remark}
    As already specified in the main text of the paper, the results claimed in Theorem~\ref{th:Complexity} are provided with respect to the \emph{data complexity} measure, \ie the complexity is with respect to the size of the database $D$ (and also the sets $E$ and $V$ for those decision problems that require it).
\end{remark}

\Complexity*

\begin{proof}
    \textbf{Lower bounds:} All the lower bounds immediately follow from~\cite{BienvenuPODS22}, and therefore they hold already for specifications without rules for values. Furthermore, as already observed in the main text of the paper, due to Theorem~\ref{simulation}, all the lower bounds hold even for specifications without rules for objects. 

    \textbf{Upper bounds:} We now provide all the upper bounds claimed in the statement of the theorem.

    \begin{center}
        \underline{\recpb is in $\PTIME$.}
    \end{center}
    \begin{proof}
        Given an ER specification $\E=\tup{\hsrulesO,\hsrulesV,\Delta}$ over a schema $\S$, an $\S$-database $D$, a binary relation $E$ on $\ObjD(D)$, and a binary relation $V$ on $\StrPos(D)$, following Definition~\ref{def:SOL}, we have that $\tup{E,V} \in \SOL(D,\Sigma)$ if and only if \myi $\tup{E,V}$ is a candidate solution for $(D,\Sigma)$, and \myii $\extdat{D}{E}{V}\models \HRO$, $\extdat{D}{E}{V} \models \HRV$, and $\extdat{D}{E}{V} \models \Delta$. 
    
        To check condition \myii, it is sufficient to check that $\extdat{D}{E}{V} \models \hro$ for each $\hro \in \HRO$, $\extdat{D}{E}{V} \models \hrv$ for each $\hrv \in \HRV$, and that $\extdat{D}{E}{V} \models \delta$ for each $\delta \in \Delta$. By definition, checking whether $\extdat{D}{E}{V} \models \hro$ for a hard rule $\hro=q(x,y) \harrow \linkO(x,y)$ amounts to check that each pair of object constants $(c,c')$ with $(c,c') \in q(\extdat{D}{E}{V})$ is such that $(c,c') \in E$. Thus, we can pick all possible pairs of object constants $(c,c')$ occurring in $\ObjD(D)$ and check the following: if $(c,c') \in q(\extdat{D}{E}{V})$, then $(c,c') \in E$. Since, due to Lemma~\ref{lem:Eval}, we know that it is possible to verify whether $\extdat{D}{E}{V} \models q[c,c']$, i.e. whether $(c,c') \in q(\extdat{D}{E}{V})$, in polynomial time in the size of $D$, $E$, and $V$, we immediately derive that checking whether $\extdat{D}{E}{V} \models \hro$ for a $\hro \in \HRO$ can be done in polynomial time in the size of $D$, $E$, and $V$. It follows that checking whether $\extdat{D}{E}{V} \models \HRO$ can be done in polynomial time in the size of $D$, $E$, and $V$. Similar considerations apply for checking whether $\extdat{D}{E}{V} \models \HRV$ and whether $\extdat{D}{E}{V} \models \Delta$.
    
        We conclude by showing that checking condition \myi, \ie $\tup{E,V}$ is a candidate solution for $(D,\Sigma)$, can be done in polynomial time in the size of $D$, $E$, and $V$ as well. To this aim, we can start with $E'\coloneqq \eqrel(\emptyset,\ObjD(D))$ and $V' \coloneqq \eqrel(\emptyset,\StrPos(D))$ and repeat the following step until a fixpoint is reached: if there is some pair $(o,o') \in E$ (\resp $(\tup{t,i},\tup{t',i'}) \in V$) such that $(o,o') \in q(\extdat{D}{E}{V})$ (\resp $(t,t') \in q(\extdat{D}{E}{V})$) for some $q(x,y) \rightarrow \linkO(x,y) \in \hsrulesO$ (\resp $q(x_t,y_t) \rightarrow \linkV(\tup{x_t,i},\tup{y_t,i'}) \in \hsrulesV$) and $(o,o') \not \in E'$ (\resp $(\tup{t,i},\tup{t',i'}) \not \in V'$), then set $E'\coloneqq\eqrel(E' \cup \{(o,o')\},\ObjD(D))$ (\resp $V'\coloneqq\eqrel(V' \cup \{(\tup{t,i},\tup{t',i'})\},\StrPos(D))$). Once the fixpoint is reached, we can check whether the given $E$ and $V$ coincides with the computed $E'$ and $V'$, respectively. It is easy to verify that the given $\tup{E,V}$ is a candidate solution for $(D,\Sigma)$ if and only if $E$ and $V$ coincide with the computed $E'$ and $V'$, respectively, and that checking this, as well as computing $E'$ and $V'$ as described above, can be done in polynomial time in the size of $D$, $E$, and $V$. From the above considerations, it is immediate to derive a deterministic procedure that checks whether $\tup{E,V} \in \SOL(D,\E)$ and that runs in polynomial time in the size of $D$, $E$, and $V$. 
    \end{proof}

    \begin{center}
        \underline{\existpb is in $\NP$.}
    \end{center}
    \begin{proof}
        Given an ER specification $\E=\tup{\hsrulesO,\hsrulesV,\Delta}$ over a schema $\S$ and an $\S$-database $D$, we first guess a pair $\tup{E,V}$, where $E$ is a set of pairs of object constants from $\ObjD(D)$ and $V$ is a set of pairs of cells from $\StrPos(D)$. Then, in polynomial time in the size of $D$, $E$, and $V$ (and therefore, in polynomial time in the size of $D$ as well because both the guessed $E$ and $V$ are polynomially related to $D$) we check whether $\tup{E,V} \in \SOL(D,\E)$ using exactly the same technique illustrated above in the upper bound provided for \recpb. So, checking whether $\SOL(D,\Sigma) \neq \emptyset$ can be done in $\NP$ in the size of $D$.
    \end{proof}

    \begin{center}
        \underline{\maxrecpb is in $\coNP$.}
    \end{center}
    \begin{proof}
        Given an ER specification $\E=\tup{\hsrulesO,\hsrulesV,\Delta}$ over a schema $\S$, an $\S$-database $D$, a binary relation $E$ on $\ObjD(D)$, and a binary relation $V$ on $\StrPos(D)$, we now show how to check whether $\tup{E,V} \not \in \MaxSOL(D,\E)$ in $\NP$ in the size of $D$, $E$, and $V$. First, note that, by definition, $\tup{E,V} \not \in \MaxSOL(D,\E)$ if and only if either $\tup{E,V} \not \in \SOL(D,\E)$ or there exists $\tup{E'',V''}$ such that $\tup{E'',V''} \in \SOL(D,\E)$ and $E \cup V \subsetneq E'' \cup V''$.

        We first guess a pair $\tup{E',V'}$, where $E'$ is a set of pairs of object constants from  $\ObjD(D)$ and $V'$ is a set of pairs of cells from $\StrPos(D)$. Then, we check whether \myi $\tup{E,V} \not \in \SOL(D,\E)$ or \myii $\tup{E \cup E',V \cup V'} \in \SOL(D,E)$ and there is at least one $\alpha$ (either a pair of object constants or a pair of cells) such that $\alpha \in E' \cup V'$ but $\alpha \not \in E \cup V$. If either condition \myi or condition \myii is satisfied, then we return \true; otherwise, we return \false. Correctness of the above procedure for deciding $\tup{E,V} \not \in \MaxSOL(D,\E)$ is trivial.
    
        Due to the upper bound provided above for \recpb, observe that \myi checking whether $\tup{E,V} \not \in \SOL(D,\E)$ can be done in polynomial time in the size of $D$, $E$, and $V$, and, moreover, \myii checking whether $\tup{E \cup E',V \cup V'} \in \SOL(D,E)$ can be done in polynomial time in the size of $D$, $E \cup E'$, and $V \cup V'$ (and therefore, in polynomial time in the size of $D$, $E$, and $V$ as well because both the guessed $E'$ and $V'$ are polynomially related to $D$). So, checking whether $\tup{E,V} \not \in \MaxSOL(D,\E)$ can be done in $\NP$ in the size of $D$, $E$, and $V$. 
    \end{proof}

    \begin{center}
        \underline{\pmergepb is in $\NP$.}
    \end{center}
    \begin{proof}
        Given an ER specification $\E=\tup{\hsrulesO,\hsrulesV,\Delta}$ over a schema $\S$, an $\S$-database $D$, and a candidate merge $\alpha$ (either a pair of object constants of the form $(o,o')$ or a pair of cells of the form $(\tup{t,i},\tup{t',i'})$), we first guess a pair $\tup{E,V}$ such that $\alpha \in E \cup V$, where $E$ is a set of pairs of object constants from $\ObjD(D)$ and $V$ is a set of pairs of cells from $\StrPos(D)$. Then, in polynomial time in the size of $D$, $E$, and $V$ (and therefore, in polynomial time in the size of $D$ as well because both the guessed $E$ and $V$ are polynomially related to $D$) we check whether $\tup{E,V} \in \SOL(D,\E)$ using exactly the same technique illustrated above in the upper bound provided for \recpb. If $\tup{E,V} \in \SOL(D,\E)$, then we return \true; otherwise, we return \false. 
        
        Correctness of the above procedure is guaranteed by the next observation which immediately follows by the definition of maximal solutions (Definition~\ref{def:maxsol}): if a pair $\alpha$ is such that $\alpha \in E \cup V$ for some $\tup{E,V} \in \SOL(D,\E)$, then either \myi $\tup{E,V} \in \MaxSOL(D,\E)$ or \myii there exists at least one $\tup{E',V'}$ such that $E \cup V \subsetneq E' \cup V'$ and $\tup{E',V'} \in \MaxSOL(D,\E)$. In both cases, we will have that $\alpha \in E'' \cup V''$ for some $\tup{E'',V''} \in \MaxSOL(D,\E)$. So, checking whether a candidate merge $\alpha$ belongs to $E \cup V$ for some $\tup{E,V} \in \MaxSOL(D,\E)$ can be done in $\NP$ in the size of $D$.
    \end{proof}

    \begin{center}
        \underline{\cmergepb is in $\PItwop$.}
    \end{center}
    \begin{proof}
        Given an ER specification $\E=\tup{\hsrulesO,\hsrulesV,\Delta}$ over a schema $\S$, an $\S$-database $D$, and a candidate merge $\alpha$ (either a pair of object constants of the form $(o,o')$ or a pair of cells of the form $(\tup{t,i},\tup{t',i'})$), we now show how to check whether $\alpha \not \in E \cup V$ for some $\tup{E,V} \in \MaxSOL(D,\E)$ in $\SItwop$ in the size of $D$. 

        We first guess a pair $\tup{E,V}$ such that $\alpha \not\in E \cup V$, where $E$ is a set of pairs of object constants from $\ObjD(D)$ and $V$ is a set of pairs of cells from $\StrPos(D)$. Then, we check whether \myi $\SOL(D,\E)=\emptyset$ or \myii $\tup{E,V} \in \MaxSOL(D,\E)$. If either condition \myi or condition \myii is satisfied, then we return \true; otherwise, we return \false. Correctness of the above procedure for deciding whether $\alpha \not\in E \cup V$ for some $\tup{E,V} \in \MaxSOL(D,\E)$ is trivial.

        Observe that condition \myi can be checked by means of a $\coNP$-oracle, since due to the upper bound provided above for \existpb checking whether $\SOL(D,\E)=\emptyset$ can be done in $\coNP$ in the size of $D$. As for condition \myii, it is possible to check whether $\tup{E,V} \in \MaxSOL(D,\E)$ again by means of a $\coNP$-oracle, since due to the upper bound provided above for \maxrecpb checking whether $\tup{E,V} \in \MaxSOL(D,\E)$ can be done in $\coNP$ in the size of $D$, $E$, and $V$ (and therefore, in $\coNP$ in the size of $D$ as well because both the guessed $E$ and $V$ are polynomially related to $D$). So, overall, checking whether a candidate merge $\alpha$ is such that $\alpha \not \in E \cup V$ for some $\tup{E,V} \in \MaxSOL(D,\E)$ can be done in $\SItwop$ in the size of $D$.
    \end{proof}

    \begin{center}
        \underline{\possanspb is in $\NP$.}
    \end{center}
    \begin{proof}
        Given an ER specification $\E=\tup{\hsrulesO,\hsrulesV,\Delta}$ over a schema $\S$, an $\S$-database $D$, an $n$-ary CQ $q(\vec{x})$ over $\S$, and an $n$-tuple $\vec{c}$ of constants, we first guess a pair $\tup{E,V}$, where $E$ is a set of pairs of object constants from $\ObjD(D)$ and $V$ is a set of pairs of cells from $\StrPos(D)$. Then, we check whether \myi $\tup{E,V} \in \SOL(D,\E)$ and \myii $\vec{c} \in q(\extdat{D}{E}{V})$, \ie $\extdat{D}{E}{V} \models q[\vec{c}]$. If both conditions \myi and \myii are satisfied, then we return \true; otherwise, we return \false. 
        
        Correctness of the above procedure for deciding whether $\vec{c} \in q(\extdat{D}{E}{V})$ for some $\tup{E,V} \in \MaxSOL(D,\E)$ is guaranteed by Lemma~\ref{lem:Preservation}.

        Observe that checking whether $\tup{E,V} \in \SOL(D,\E)$ can be done in polynomial time in the size of $D$, $E$, and $V$ (and therefore, in polynomial time in the size of $D$ as well because both the guessed $E$ and $V$ are polynomially related to $D$) using exactly the same technique illustrated above in the upper bound provided for \recpb. Finally, due to Lemma~\ref{lem:Eval}, checking whether $\extdat{D}{E}{V} \models q[\vec{c}]$ can be done in polynomial time in the size of $D$, $E$, and $V$ (and therefore, in polynomial time in the size of $D$ as well because both the guessed $E$ and $V$ are polynomially related to $D$) as well. So, overall, checking whether $\vec{c} \in q(\extdat{D}{E}{V})$ for some $\tup{E,V} \in \MaxSOL(D,\E)$ can be done in $\NP$ in the size of $D$, as required.
    \end{proof}

    \begin{center}
        \underline{\certanspb is in $\PItwop$.}
    \end{center}
    \begin{proof}
        Given an ER specification $\E=\tup{\hsrulesO,\hsrulesV,\Delta}$ over a schema $\S$, an $\S$-database $D$, an $n$-ary CQ $q(\vec{x})$ over $\S$, and an $n$-tuple $\vec{c}$ of constants, we now show how to check whether $\vec{c} \not \in q(\extdat{D}{E}{V})$ for some $\tup{E,V} \in \MaxSOL(D,\E)$ in $\SItwop$ in the size of $D$.

        We first guess a pair $\tup{E,V}$, where $E$ is a set of pairs of object constants from $\ObjD(D)$ and $V$ is a set of pairs of cells from $\StrPos(D)$. Then, we check whether \myi $\SOL(D,\E)=\emptyset$ or \myii $\tup{E,V} \in \MaxSOL(D,\E)$ and $\vec{c} \not \in q(\extdat{D}{E}{V})$, \ie $\extdat{D}{E}{V} \not\models q[\vec{c}]$. If either condition \myi or condition \myii is satisfied, then we return \true; otherwise, we return \false. Correctness of the above procedure for deciding $\vec{c} \not \in q(\extdat{D}{E}{V})$ for some $\tup{E,V} \in \MaxSOL(D,\E)$ is trivial.

        Observe that condition \myi can be checked by means of a $\coNP$-oracle, since due to the upper bound provided above for \existpb checking whether $\SOL(D,\E)=\emptyset$ can be done in $\coNP$ in the size of $D$. As for condition \myii, it is possible to check whether $\tup{E,V} \in \MaxSOL(D,\E)$ again by means of a $\coNP$-oracle, since due to the upper bound provided above for \maxrecpb checking whether $\tup{E,V} \in \MaxSOL(D,\E)$ can be done in $\coNP$ in the size of $D$, $E$, and $V$ (and therefore, in $\coNP$ in the size of $D$ as well because both the guessed $E$ and $V$ are polynomially related to $D$). Finally, due to Lemma~\ref{lem:Eval}, checking whether $\extdat{D}{E}{V} \not\models q[\vec{c}]$ can be done in polynomial time in the size of $D$, $E$, and $V$ (and therefore, in polynomial time in the size of $D$ as well because both the guessed $E$ and $V$ are polynomially related to $D$) as well. So, overall, checking whether $\vec{c} \not \in q(\extdat{D}{E}{V})$ for some $\tup{E,V} \in \MaxSOL(D,\E)$ can be done in $\SItwop$ in the size of $D$.
    \end{proof}

    \begin{center}
        \underline{For specifications that do not use inequality atoms in denial constraints, \existpb is in $\PTIME$.}
    \end{center}
    \begin{proof}
        Given an ER specification $\E=\tup{\hsrulesO,\hsrulesV,\Delta}$ over a schema $\S$ such that all the denial constraints in $\Delta$ have no inequality atoms and an $\S$-database $D$, we start with $E\coloneqq \eqrel(\emptyset,\ObjD(D))$ and $V\coloneqq \eqrel(\emptyset,\StrPos(D))$ and repeat the following step until a fixpoint is reached: if there is some pair $(o,o')$ (\resp $(\tup{t,i},\tup{t',i'})$) such that $(o,o') \in q(\extdat{D}{E}{V})$ (\resp $(t,t') \in q(\extdat{D}{E}{V})$) for some $q(x,y) \harrow \linkO(x,y) \in \HRO$ (\resp $q(x_t,y_t) \harrow \linkV(\tup{x_t,i},\tup{y_t,i'}) \in \HRV$) and $(o,o') \not \in E'$ (\resp $(\tup{t,i},\tup{t',i'}) \not \in V'$), then set $E'\coloneqq\eqrel(E' \cup \{(o,o')\},\ObjD(D))$ (\resp $V'\coloneqq\eqrel(V' \cup \{(\tup{t,i},\tup{t',i'})\},\StrPos(D))$). Once the fixpoint is reached, we can check whether the obtained $\tup{E,V}$ is a solution for $(D,\E)$ by simply checking whether $\extdat{D}{E}{V} \models \Delta$ (by construction, we already have that $\extdat{D}{E}{V} \models \HRO$, $\extdat{D}{E}{V} \models \HRV$, and $\tup{E,V}$ is a candidate solution for $(D,\E)$). If it is the case that $\extdat{D}{E}{V} \models \Delta$, then we return \true; otherwise, we return \false. 
        
        Clearly, due to Lemma~\ref{lem:Eval}, the above procedure can be carried out in polynomial time in the size of $D$. The correctness is guaranteed by the fact that, on the one hand, the computed $E$ (\resp $V$) will contain only those pairs of object constants (pairs of cells) that are necessary to satisfy the hard rules, and, on the other hand, by Lemma~\ref{lem:NoMoreMerg}, if $\extdat{D}{E}{V} \not\models \Delta$, then there is no way to add other pairs to $E \cup V$ in order to satisfy $\Delta$.
    \end{proof}

    \begin{center}
        \underline{For specifications that do not use inequality atoms in denial constraints, \maxrecpb is in $\PTIME$.}
    \end{center}
    \begin{proof}
        Given an ER specification $\E=\tup{\hsrulesO,\hsrulesV,\Delta}$ over a schema $\S$ such that all the denial constraints in $\Delta$ have no inequality atoms, an $\S$-database $D$, a binary relation $E$ on $\ObjD(D)$, and a binary relation $V$ on $\StrPos(D)$, as a first step we check whether $\tup{E,V} \in \SOL(D,\E)$ using exactly the same technique illustrated above in the upper bound provided for \recpb. If the given $\tup{E,V}$ is not a solution for $(D,\E)$, then we return \false; otherwise, we continue as follows. We collect in a set $S$ all those pairs $\alpha$ of object constants $(o,o')$ and cells $(\tup{t,i},\tup{t',i'})$ such that $\alpha \not \in E \cup V$ and either there is a soft rule for objects in $\SRO$ of the form $q(x,y)\sarrow \linkO(x,y)$ with $(o,o') \in q(\extdat{D}{E}{V})$ or there is a soft rule for values in $\SRV$ of the form $q(x_t,y_t) \sarrow \linkV(\tup{x_t,i},\tup{y_t,i'})$ with $(t,t') \in q(\extdat{D}{E}{V})$.
        
        Then, for each possible $\alpha \in S$, we do the following. If $\alpha$ is a pair of object constants, then we start with $E'\coloneqq \eqrel(E \cup \{\alpha\},\ObjD(D))$ and $V' \coloneqq V$; otherwise (\ie $\alpha$ is a pairs of cells), we start with $E' \coloneqq E$ and $V' \coloneqq \eqrel(V \cup \{\alpha\},\StrPos(D))$. In both cases, we repeat the following until a fixpoint is reached: if there is some pair $(o,o')$ (\resp $(\tup{t,i},\tup{t',i'})$) such that $(o,o') \in q(\extdat{D}{E}{V})$ (\resp $(t,t') \in q(\extdat{D}{E}{V})$) for some $q(x,y) \harrow \linkO(x,y) \in \HRO$ (\resp $q(x_t,y_t) \harrow \linkV(\tup{x_t,i},\tup{y_t,i'}) \in \HRV$) and $(o,o') \not \in E'$ (\resp $(\tup{t,i},\tup{t',i'}) \not \in V'$), then set $E'\coloneqq\eqrel(E' \cup \{(o,o')\},\ObjD(D))$ (\resp $V'\coloneqq\eqrel(V' \cup \{(\tup{t,i},\tup{t',i'})\},\StrPos(D))$). Once the fixpoint is reached, we check whether the obtained $\tup{E',V'}$ is such that $\extdat{D}{E'}{V'} \models \Delta$. If this is the case for some $\alpha \in S$, then we return \false; otherwise, we return \true. 

        Clearly, due to Lemma~\ref{lem:Eval}, the above procedure can be carried out in polynomial time in the size of $D$, $E$, and $V$, since collecting those pairs in $S$ can be done in polynomial time in the size of $D$, $E$, and $V$. The correctness is guaranteed by the following observation. Given a solution $\tup{E,V}$ for $(D,\E)$, to establish that $\tup{E,V} \not \in \MaxSOL(D,\E)$ it is sufficient to apply a ``minimal'' extend to $\tup{E,V}$ and see whether such minimal extend leads to a solution for $(D,\E)$. More precisely, the minimal extend consists in either adding to $E$ a pair of object constants $\alpha \not \in E$ due to some soft rule for objects or adding to $V$ a pair of cells $\alpha \not \in V$ due to some soft rule for values (since $\tup{E,V} \in \SOL(D,\E)$, and thus $\extdat{D}{E}{V} \models \HRO$ and $\extdat{D}{E}{V} \models \HRV$, the minimal extend cannot be due to the hard rules), then compute a new pair $\tup{E',V'}$ based on $E \cup V \cup \{\alpha\}$ that satisfy all the hard rules (the added $\alpha$ can now trigger other merges that must be included to satisfy the hard rules), and finally check that the obtained $\tup{E',V'}$ is such that $\extdat{D}{E'}{V'} \models \Delta$. If for no such $\alpha$ this is the case, \ie every time we try to extend $E$ or $V$ with some $\alpha \in S$ we end up with an $\tup{E',V'}$ such that $\extdat{D}{E'}{V'} \not \models \Delta$, then, due to Lemma~\ref{lem:NoMoreMerg}, we get that every candidate solution $\tup{E'',V''}$ for $(D,\E)$ for which \myi $E \cup V \subsetneq E'' \cup V''$, \myii $\extdat{D}{E''}{V''} \models \HRO$, and \myiii $\extdat{D}{E''}{V''} \models \HRV$ is be such that $\extdat{D}{E''}{V''} \not \models \Delta$.
    \end{proof}

    \begin{center}
        \underline{For specifications that do not use inequality atoms in denial constraints, \cmergepb is in $\coNP$.}
    \end{center}
    \begin{proof}
        Given an ER specification $\E=\tup{\hsrulesO,\hsrulesV,\Delta}$ over a schema $\S$ such that all the denial constraints in $\Delta$ have no inequality atoms, an $\S$-database $D$, and a candidate merge $\alpha$ (either a pair of object constants of the form $(o,o')$ or a pair of cells of the form $(\tup{t,i},\tup{t',i'})$), we now show how to check whether $\alpha \not \in E \cup V$ for some $\tup{E,V} \in \MaxSOL(D,\E)$ in $\NP$ in the size of $D$. 

        We first guess a pair $\tup{E,V}$ such that $\alpha \not\in E \cup V$, where $E$ is a set of pairs of object constants from $\ObjD(D)$ and $V$ is a set of pairs of cells from $\StrPos(D)$. Then, we check whether \myi $\SOL(D,\E)=\emptyset$ or \myii $\tup{E,V} \in \MaxSOL(D,\E)$. If either condition \myi or condition \myii is satisfied, then we return \true; otherwise, we return \false. Correctness of the above procedure for deciding whether $\alpha \not\in E \cup V$ for some $\tup{E,V} \in \MaxSOL(D,\E)$ is trivial.

        Observe that, since all the denial constraints in $\Delta$ have no inequality atoms, as already shown in the upper bound provided for \existpb in the case of specifications that do not use inequality atoms in denial constraints, condition \myi can be checked in polynomial time in the size of $D$, $E$, and $V$ (and therefore, in polynomial time in the size of $D$ as well because both the guessed $E$ and $V$ are polynomially related to $D$). Furthermore, as already shown in the upper bound provided for \maxrecpb in the case of specifications that do not use inequality atoms in denial constraints, checking whether $\tup{E,V} \in \MaxSOL(D,\E)$ can be done in polynomial time in the size of $D$, $E$, and $V$ (and therefore, in polynomial time in the size of $D$ as well because both the guessed $E$ and $V$ are polynomially related to $D$). 
        So, overall, checking whether a candidate merge $\alpha$ is such that $\alpha \not \in E \cup V$ for some $\tup{E,V} \in \MaxSOL(D,\E)$ can be done in $\NP$ in the size of $D$.
    \end{proof}

    \begin{center}
        \underline{For specifications that do not use inequality atoms in denial constraints, \certanspb is in $\coNP$.}
    \end{center}
    \begin{proof}
        Given an ER specification $\E=\tup{\hsrulesO,\hsrulesV,\Delta}$ over a schema $\S$ such that all the denial constraints in $\Delta$ have no inequality atoms, an $\S$-database $D$, an $n$-ary CQ $q(\vec{x})$ over $\S$, and an $n$-tuple $\vec{c}$ of constants, we now show how to check whether $\vec{c} \not \in q(\extdat{D}{E}{V})$ for some $\tup{E,V} \in \MaxSOL(D,\E)$ in $\NP$ in the size of $D$.

        We first guess a pair $\tup{E,V}$, where $E$ is a set of pairs of object constants from $\ObjD(D)$ and $V$ is a set of pairs of cells from $\StrPos(D)$. Then, we check whether \myi $\SOL(D,\E)=\emptyset$ or \myii $\tup{E,V} \in \MaxSOL(D,\E)$ and $\vec{c} \not \in q(\extdat{D}{E}{V})$, \ie $\extdat{D}{E}{V} \not\models q[\vec{c}]$. If either condition \myi or condition \myii is satisfied, then we return \true; otherwise, we return \false. Correctness of the above procedure for deciding whether $\vec{c} \not \in q(\extdat{D}{E}{V})$ for some $\tup{E,V} \in \MaxSOL(D,\E)$ is trivial.

        Observe that, since all the denial constraints in $\Delta$ have no inequality atoms, as already shown in the upper bound provided for \existpb in the case of specifications that do not use inequality atoms in denial constraints, condition \myi can be checked in polynomial time in the size of $D$, $E$, and $V$ (and therefore, in polynomial time in the size of $D$ as well because both the guessed $E$ and $V$ are polynomially related to $D$). Furthermore, as already shown in the upper bound provided for \maxrecpb in the case of specifications that do not use inequality atoms in denial constraints, checking whether $\tup{E,V} \in \MaxSOL(D,\E)$ can be done in polynomial time in the size of $D$, $E$, and $V$ (and therefore, in polynomial time in the size of $D$ as well because both the guessed $E$ and $V$ are polynomially related to $D$). Finally, due to Lemma~\ref{lem:Eval}, checking whether $\extdat{D}{E}{V} \not\models q[\vec{c}]$ can be done in polynomial time in the size of $D$, $E$, and $V$ (and therefore, in polynomial time in the size of $D$ as well because both the guessed $E$ and $V$ are polynomially related to $D$) as well. So, overall, checking whether $\vec{c} \not \in q(\extdat{D}{E}{V})$ for some $\tup{E,V} \in \MaxSOL(D,\E)$ can be done in $\NP$ in the size of $D$.
    \end{proof}
    \phantom\qedhere
\end{proof}
\newcommand{\NeqO}{\textit{NEqO}}
\newcommand{\NeqV}{\textit{NEqV}}
\newcommand{\ActiveO}{\textit{ActiveO}}
\newcommand{\ActiveV}{\textit{ActiveV}}

\thmASPSols*
%
\begin{proof}
   
Let $D$ be a database and $\E=\tup{\hsrulesO,\hsrulesV,\Delta}$ a specification.
We define a normal logic program \PSol whose answer sets can be used to generate solutions for $(D. \Sigma)$. 
First, we have rules  
$$
\begin{aligned}
\proj(u_t,i,v) \leftarrow&  R(u_t, v_1,\dots, v_{i-i}, v, v_{i+1},\dots, v_k )   \qquad(R/k \in \S, 1 \leq i \leq k ) \\
\obj(o)  \leftarrow & R(u_t, v_1,\dots, v_{i-i}, o, v_{i+1},\dots, v_k )   \qquad(R/k \in \S, 1 \leq i \leq k  \text{ and } \typef(R,j)= \Obj ) 
\end{aligned}
$$
such that ground instances of $\proj(t,i,c)$ of $\proj/3$ encode that $t[i]=c$, and $\obj/1$ stores the set of object constants in $D$.

We start by describing the translation of rules in $\hsrulesV$ 
For every CQ $q(x_t,y_t)$ in the body of a rule in $\hsrulesV$, let $\hat{q}(x_t,y_t)$ be constructed by 
replacing each occurrence of a non-distinguished variable in $q$ (excluding similarity atoms) with a fresh variable. We will write  $\occ(v)$ (resp.\ \Oocc) to denote the set of pairs $(u_t,i)$ such that $v$ occurs in a value (resp.\ object ) position $i$ in an atom $R(u_t,v_1, \dots, v_k)$ in $q$, and $v_{(u_t,i)}$ to denote the fresh variable introduced for replacing occurrence $(u_t,i)$ of $v$. After the renaming, we proceed as follows:

\begin{enumerate}
\item  for every join variable $v$ from $q$ occurring in some value position, take a pair of fresh variables $u'_t,k$ and add to
$\hat{q}$ the set of atoms $$\{EqV(u_t, i, u'_t, k) \mid (u_t,i) \in \occ(v) \} $$
 \item for every join variable $v$ from $q$ from $q$ occurring in some object position, take a fresh variable $v'$ and add to $\hat{q}$ the set of atoms $$\{\EqO(v_{(u_t,i)},v') \mid (u_t,i) \in \Oocc(v) \};$$ 
 \item for each similarity atom $v \approx w$, take a fresh variables $v',w'$ and replace that atom by the set of atoms
 $$\{\val(u_t, i, v' )| (u_t,i) \in \occ(v) \} \cup \{\val(u'_t, j, w' )| (u'_t,j) \in \occ(w) \} \cup \{v' \approx w'\}$$
 where $\val$
 is a predicate is defined by the rule:
 $$\val(u_t,i,v) \leftarrow \EqV(u_t,i,u'_t,j), \proj(u'_t,j,v)$$
 \end{enumerate}
Rules in $\hsrulesV$ are then translated as follows.
Every hard rule $q(x_t,y_t) \harrow\, \linkV(\tup{x_t,i},\tup{y_t,j})$ is translated into the rule  
\begin{equation}
\EqV(x_t,i,y_t,j) \leftarrow \hat{q}(x_t,y_t).
\end{equation}
Every soft  rule $q(x_t,y_t) \sarrow \, \linkV(\tup{x_t,i},\tup{y_t,j}) \in \hsrulesV$ is translated into 
\begin{equation}
 \ActiveV(x_t,i,y_t,j) \leftarrow \hat{q}(x_t,y_t) 
\end{equation}
To capture the merge choices provided by soft rules, we have the following rules:
\begin{equation}
    \begin{aligned}
        \EqV(x_t,i, y_t,j) \leftarrow & not \ \NeqV(x_t, i, y_t, j)\\
        \NeqV(x_t, i, y_t, j) \leftarrow & \ActiveV(x_t, i, y_t, j) not \ \EqV(x_t, i, y_t, j)
    \end{aligned}
\end{equation}
Rules in $\hsrulesO$ are treated analogously by constructing a set of atoms $q^+$ for every CQ $q(x,y)$ in the body of a rule in $\hsrulesO$ similarly to the construction of $\hat{q}$ described above, except that we initially rename \emph{all} variables in  $q(x,y)$ using fresh variables, and after step 3, we pick a single variable $x_{(u_t,i)}$ and a single variable $y_{(u'_t,j)}$, and replaced them with $x$ and $y$, respectively. 
Each hard rule $q(x,y) \harrow \linkO(x,y)$ is translated into the rule
\begin{equation}
    \EqO(x,y) \leftarrow q^+(x,y). 
\end{equation}
And each soft rule $q(x,t) \sarrow \linkO(x,y)$ is translated into the rule
\begin{equation}
    \ActiveO(x,y) \leftarrow q^+(x,y). 
\end{equation}
Similarly as above, to capture the choice of merges given by soft rules, we have the rules
\begin{equation}
    \begin{aligned}
        \EqO(x,y) \leftarrow & \ActiveO(x,y), not \ \NeqO(x,y) \\ 
        \NeqO(x,t) \leftarrow & \ActiveO(x,y), not \ \EqO(x,y)
    \end{aligned}
\end{equation}
For each denial constraint $\exists \vec{y} \per \varphi(\vec{y}) \rightarrow \bot \in \Delta$, the program  \PSol contains a constraint
\begin{equation}
 \leftarrow \hat{\varphi}   
\end{equation}
where $\hat{\varphi}$ is a set of atoms, constructed from $\varphi$ in similarly to $\hat{q}$.

Additionally, \PSol contains rules enforcing  that  $\EqO$ and $\EqV$ are both  equivalence relations:
\begin{equation}\label{enc:Eqrel}
\begin{aligned}
\EqO(u,u) \leftarrow &\textit{Obj}(u) & && \EqV(u_t,i, u_t,i) \leftarrow &Tid(u_t), Valpos(u_t, i)\\
\EqO(u,v) \leftarrow &\EqO(v,u) && &\EqV(u_t,i, v_t, j) \leftarrow & \EqV(v_t,j , u_t,i)\\
\EqO(u,w) \leftarrow &\EqO(u,v), \EqO(v,w) && & \EqV(u_t,i,w_t, j) \leftarrow  & \EqV(u_t,j, v_t,k),
   \EqV(v_t,k, w_t,j)
\end{aligned}
\end{equation}

The overall proof strategy is then similar to the one for original ASP encoding of \textsc{Lace}, \cite{BienvenuPODS22}. 
Given $\tup{\merO,\merV}\in \SOL(D,\E)$, we construct a corresponding set of facts $M_E$, 
which will include $\EqO(a,b)$ for all $(a,b) \in \merO$, $\EqV(\tup{t,i},\tup{t',i'})$ for all 
$(\tup{t,i},\tup{t',i'}) \in \merV$, as well as $\ActiveO$ and $\ActiveV$ facts whenever there is a soft rule 
whose body is satisfied, and $\NeqO$ and $\NeqV$ facts when the corresponding pair is not included in $\tup{\merO,\merV}$,
plus the expected facts for the auxiliary predicates ($\proj, \obj, \ldots$).
Then we can show, by following the sequence of candidate solutions (and corresponding rule applications) that
witness $\tup{\merO,\merV}\in \SOL(D,\E)$ that 
$M_E$ is the minimal model of the reduct of the grounding of 
$\Pi_{Sol}$ w.r.t.\ $D$ relative to $M_E$, which means $M_E$ is a stable model of $(\Pi_{Sol}, D)$. 
$\Pi_{Sol}$, which means $M_E$ is a stable model of $(\Pi_{Sol}, D)$. 
Conversely, we can transform every stable model $M$ of $(\Pi_{Sol}, D)$ into a corresponding solution 
$\tup{\merO_M,\merV_M}$ where $\merO_M = \{(a,b) \mid \EqO(a,b) \in M\}$ and $\merV_M = \{(\tup{t,i},\tup{t',i'}) \mid \EqV(t,i,t',i') \in M\}$.
Essentially, we reproduce the rule applications that led to the introduction of the $\EqO$ and $\EqV$ facts in $M$ (in the minimal model of the reduct of the ground program relative to $M$)
to produce a corresponding sequence of candidate solutions that witnesses that  $\tup{\merO_M,\merV_M} \in \SOL(D,\E)$. 


While the high-level argument is quite close to the one from \cite{BienvenuPODS22}, 
it remains to argue that the handling of joins between value positions and the similarity atoms respects Definition \ref{queryeval}. 
Consider a rule for values $\rho=q(x_t,y_t) \harrow  \linkV(\tup{x_t,i},\tup{y_t,j}) $, which is encoded by an ASP rule
$\hat{\rho}$ with body $\hat{q}(x_t,y_t)$ (the argument is similar for the rules that encode rules for objects). 
Due to the renaming of variables, relational atoms in $\hat{q}(x_t,y_t)$ are mapped onto the original dataset independently, which 
mirrors how we use separate mappings $g_\ratom$ for the different relational atoms $\ratom \in q$. 
Moreover, for every join variable $v$ from $q$ occurring in some value position, $\hat{q}(x_t,y_t)$ 
contains the atoms in $\{\EqV(u_t, i, u'_t, k) \mid (u_t,i) \in \occ(v) \} $, which all importantly, all use the same pair of fresh variables $u'_t,k$. 
It follows that if the body is satisfied by mapping $(u'_t, k)$ to $(t,p)$, then cell $\tup{t,p}$ has been shown to 
belong to the same equivalence class as the cell $\tup{t_i,i}$ (corresponding to $(u_t, i)$), for every $(u_t,i) \in \occ(v)$. 
This ensures that there is a non-empty intersection between the sets of constants onto 
which the different occurrences of $v$ are mapped in the induced database corresponding to the current $(E,V)$ (which in the ASP program, 
are given by the $\EqO$ and $\EqV$ facts).  
A joined object variable $v$ is instead handled using the atoms $\{\EqO(v_{(u_t,i)},v') \mid (u_t,i) \in \Oocc(v) \}$, which just as in \cite{BienvenuPODS22}, permits joins between object constants from the same equivalence class (here given by $\EqO$ facts). 
Now take some similarity atom $v \approx w \in q$. 
Then $\hat{q}(x_t,y_t)$ contains the atoms 
 $\{\val(u_t, i, v' )| (u_t,i) \in \occ(v) \} \cup \{\val(u'_t, j, w' )| (u'_t,j) \in \occ(w) \} \cup \{v' \approx w'\}$, where $v',w'$ are fresh variables introduced for this particular similarity atom. The intuition is similar to the one for joined value variables: we use 
$\val(u_t, i, v' )$ and  $\val(u'_t, j, w' )$ to pick values that belong to the intersection of the
sets of constants onto which the different occurrences of $v$ and $w$ are mapped in the induced database. 
Then if $v', w'$ map to $c_v, c_w$, we will have $c_v \approx c_w$, which ensures that the original similarity atom $v \approx w \in q$
is satisfied in the induced database corresponding to the current $(E,V)$. 

We have thus argued that whenever we are able to satisfy the body $\hat{q}(x_t,y_t)$  of $\hat(\rho)$ by mapping $(x_t,y_t)$ to $(t,t')$, given the current set of 
$\EqO$ and $\EqV$ facts,
then we can similarly show that $\rho=q(x_t,y_t) \harrow  \linkV(\tup{x_t,i},\tup{y_t,j}) $ is applicable in the corresponding induced database
and we can add $\rho=q(x_t,y_t) \harrow  \linkV(\tup{x_t,i},\tup{y_t,j}) $. Conversely, it is not hard to see that any application of $\rho$ w.r.t.\ $D_{E,V}$
enables a corresponding application of $\hat{\rho}$ when the $\EqO$ and $\EqV$ facts contain all facts corresponding to pairs in $E,V$. 
\end{proof}

\end{document}